\documentclass[10pt]{article}
\usepackage{amsfonts}
\usepackage{amsmath}
\usepackage{amssymb}
\usepackage{amsthm}
\usepackage{thmtools}
\usepackage[pdftex]{graphicx}
\usepackage[hmargin=1in,vmargin=1in]{geometry}
\usepackage{natbib}
\usepackage{setspace}
\usepackage{enumerate}
\usepackage{comment}
\usepackage[toc,title,titletoc,header]{appendix}
\usepackage[colorlinks,citecolor=blue]{hyperref}
\usepackage{ragged2e}
\usepackage{etoolbox} %
\usepackage{wrapfig, array, booktabs, multirow, float, graphicx, tabularx, colortbl, hyperref, ulem, xcolor}
\usepackage[absolute,overlay]{textpos}
\usepackage{amsmath, amsthm, amssymb, setspace}
\usepackage{tikz}
\usetikzlibrary{shapes,arrows}

\def\qed{\rule{2mm}{2mm}}
\def\independent{\perp \!\!\! \perp}

\usepackage[pagewise,mathlines]{lineno}
\mathchardef\dash="2D

\newtheorem{theorem}{Theorem}[section]

\newtheorem{definition}{Definition}[section]

\theoremstyle{definition}
\newtheorem{example}{Example}[section]
\newtheorem{remark}{Remark}[section]
\newtheorem{assumption}{Assumption}[section]

\renewcommand\thmcontinues[1]{Continued}

\AtEndEnvironment{remark}{~\qed}
\AtEndEnvironment{example}{~\qed}

\newcommand{\pto}{\overset{P}{\rightarrow}}
\newcommand{\dto}{\overset{d}{\rightarrow}}

\newcommand{\Larrow}{\overset{L}{\rightarrow}}
\newcommand{\Parrow}{\overset{P}{\rightarrow}}

\newcommand{\I}{\mathrm{I}}
\newcommand{\E}{\mathrm{E}}

\newcommand{\Var}{\mathrm{Var}}

\begin{document}
\title{Randomization Inference: Theory and Applications\thanks{ We thank Juri Trifonov for helpful comments and suggestions. Ritzwoller gratefully acknowledges support from the National Science Foundation under the Graduate Research Fellowship. Computational support was provided by the Data, Analytics, and Research Computing (DARC) group at the Stanford Graduate School of Business (RRI:SCR\_022938).}}

\author{David M. Ritzwoller\\
Graduate School of Business\\
Stanford University\\
\url{ritzwoll@stanford.edu}
\and
Joseph P.\ Romano \\
Departments of Economics and Statistics\\
Stanford University\\
\url{romano@stanford.edu}
\and
Azeem M.\ Shaikh\\
Department of Economics\\
University of Chicago \\
\url{amshaikh@uchicago.edu}
}

\maketitle

\vspace{-0.3in}

\begin{spacing}{1.2}
\begin{abstract}
\normalsize{We review approaches to statistical inference based on randomization. Permutation tests  are treated as an important special case. Under a certain group invariance property, referred to as the ``randomization hypothesis,'' randomization tests achieve exact control of the Type I error rate in finite samples.  Although this unequivocal precision is very appealing, the range of problems that satisfy the randomization hypothesis is somewhat limited. We show that randomization tests are often asymptotically, or approximately, valid and efficient in settings that deviate from the conditions required for finite-sample error control. When randomization tests fail to offer even asymptotic Type 1 error control, their asymptotic validity may be restored by constructing an asymptotically pivotal test statistic. Randomization tests can then provide exact error control for tests of highly structured hypotheses with good performance in a wider class of problems. We give a detailed overview of several prominent applications of randomization tests, including two-sample permutation tests, regression, and conformal inference.}
\end{abstract}
\end{spacing}


\thispagestyle{empty} 
\newpage
\setcounter{page}{1}

\section{Introduction}

Randomization as a distinct approach to statistical inference dates back to at least \cite{fisher1935design}. See \cite{box1980ra} and \cite{david2008beginnings} for historical accounts. \cite{fisher1935design} considers a simple experiment. A colleague asserts that they can discern the difference between milk poured into tea and tea poured into milk. Skeptical, Fisher prepares eight cups of tea, with four cups poured in each way. These cups are presented to the colleague in a randomized order. Consider the null hypothesis that Fisher's colleague cannot discriminate between the two preparations, and so assigns them labels randomly. There are 70 ways of dividing the eight cups of tea into two groups of four. Under the null hypothesis, the probability that Fisher's colleague correctly groups the two types of tea is $1/70$. In modern terms, if Fisher's colleague correctly groups the two types of tea, then we can reject the null hypothesis with a $p$-value equal to $1/70$.  

Fisher's exposition of this simple experiment initiated a revolution in our view of the role of randomization in experimental design, statistical methodology, and science at large \citep{salsburg}. Further classical developments were given by \cite{pitman1937significance_1, pitman1937significance_2, pitman1938significance}. A more mathematical treatment of randomization inference and the construction of randomization tests was established by \cite{lehmann1949comments}, \cite{lehmann1949theory}, and \cite{hoeffding:1952}, among others. Today, randomization tests remain an active field for methodological research, and are widely used throughout basic and applied science.

This article gives an exposition and selective review of the modern perspective on randomization tests for both applied economists and econometric theorists. Part of the allure of randomization tests is that, for some problems, they can offer exact control of the Type 1 error rate, without resorting to parametric assumptions or relying on asymptotics. As most modern econometric methods rely on asymptotics, the exactness offered by randomization inference can be comforting, if not liberating. This sentiment, however, should be regulated by some important caveats. First, the specific settings where one can apply randomization tests to achieve exact, finite-sample error control are somewhat limited. In particular, exact error control is obtained only under a condition referred to as the``randomization hypothesis."   Second, even when the randomization hypothesis holds, and so exact error control is achieved, it is generally only through asymptotics that issues like efficiency, power, and sample size requirements can be well-understood.  Finally, the use of randomization tests when the randomization hypothesis fails may lead to a lack of Type 1 error control, even in large samples. However, it will be argued that if a randomization test is based on an appropriate choice of test statistic, meaning one that is asymptotically pivotal under the null hypothesis, then consistency of the randomization test can be restored while preserving the exactness of the test when the randomization hypothesis holds.

Thus,  in problems that do not meet the requirements needed for exact inference, suitably constructed randomization tests control the Type 1 error rate asymptotically. One goal of this paper is to present the intuition for such a claim.  In principle, then, randomization tests balance exact error control for tests of highly structured hypotheses with good performance in a wider class of problems. In practice, care must be taken to meet these dual goals. Achieving this balance is the main emphasis of this article.

The methodological literature on randomization tests is extensive. We provide a limited review throughout. Several book length treatments include \cite{edgington1995randomization}, \cite{pesarin2001multivariate}, \cite{good2005permutation}, \cite{solari2009permutation}, \cite{pesarin2010permutation}, \cite{salmaso2011permutation}, and \cite{berry2019primer}. \cite{kennedy1995randomization} gives a review targeted to an econometric audience. For the most part, our discussion will center on the validity of approaches based on randomization inference for testing various null hypotheses. Classical analyses of the efficiency of randomization tests include \cite{albers1976asymptotic}, \cite{bickel2011asymptotic}, and \cite{romano:1989}. More recent treatments include \cite{berrett2021optimal}, \cite{kim2022minimax}, and \cite{dobriban2022consistency}. Our paper is similar in notation and spirit to the introductory treatment in \cite{lehmann:romano:tsh:2022}, Chapter 17. Although the presentation here does not include many proofs, we aim to provide the basic intuition for statements of results.

The applied literature that makes use of randomization tests is similarly large. For example, roughly one-third of papers published in the journal {\it Experimental Economics} in 2009 made use of the two-sample Wilcoxon test \citep{chung2016asymptotically}.\footnote{Similarly, \cite{okeh2009statistical} samples a collection of medical studies and finds that roughly one-third  make use of the two-sample Wilcoxon test.} \cite{young2019channeling} re-analyses 53 experiments from leading economics journals and finds that the results of methods based on randomization inference are more robust than parametric alternatives. More generally, permutation and randomization tests have been applied to many diverse fields, such as brain imaging \citep{maris2007nonparametric}, biology \citep{blackford2009detecting}, and genomics \citep{stranger2007population}. 

We discuss several widely encountered applications of randomization inference in applied econometrics. Further applications include the methods considered by \cite{ganong2018permutation}, \cite{rambachan2020design}, \cite{chung2021permutation}, and \cite{borusyak2023nonrandom}, among many others. For reasons of space, we do not give a dedicated treatment of randomization tests of conditional independence. See, e.g., \cite{doran2014permutation}, \cite{shah2020hardness}, \cite{berrett2020conditional}, \cite{kim2022minimax}, and \cite{kim2023conditional} for several recent contributions to the literature on this problem.

We will soon define precisely what we mean by randomization inference or a randomization test.   Indeed, the terms have been used somewhat ambiguously in the literature; see \cite{zhang2023randomization} for discussion. At a high level, randomization tests can be conducted when, under the null hypothesis, there are transformations of the data that preserve the distribution of the data. If this holds, then a null  distribution may be obtained by computing the test statistic over all transformed data sets. This meaning generalizes the familiar two-sample permutation test; i.e., the setting in which two independent i.i.d.\ samples are observed and one is interested in testing the null hypothesis that both samples come from the same distribution. If the null holds, then the true null distribution of {\it any} test statistic is the same as that of the statistic computed over a randomly permuted data set. This argument will be made precise in Section \ref{section:basics}.

\section{The Basics of Randomization Inference}\label{section:basics} 

\subsection{The Randomization Hypothesis}

The observed data $X$ are valued in a sample space ${\cal X}$. The unknown probability mechanism generating $X$  is $P$ and belongs to some set $\Omega$.  As we will see, this setup applies to the usual super-population setting as well as finite population settings. The problem is to test null hypothesis $H_0:~P \in \Omega_0 \subset \Omega$. The following assumption is called the {\it randomization hypothesis}; if it holds, then tests can be constructed based on any statistic and the resulting test will obtain exact finite-sample Type 1 error control.

\begin{definition}[The Randomization Hypothesis]
Let ${\bf G}$ be a finite group of transformations $g$ of ${\cal X}$ onto itself. The null hypothesis $H_0$ implies that the distribution of $X$ is invariant under the transformations in ${\bf G}$; i.e., for every $g$ in ${\bf G}$, $gX$ and $X$ have the same distribution whenever $X$ has distribution $P$ in $\Omega_0$.
\end{definition}

\begin{example}[Sign Changes and the One-sample Problem]   \label{example:1}
Suppose that we collect data on an economic outcome before and after a policy change for a cross-section of households. For each household $i$, let $X_i$ denote the observed change in the outcome. Let $X=X^{(n)}$ collect the independent observations $X_1 , \ldots , X_n$. Each observation $X_i$ has distribution $P$  on {\bf R}. Let $\mu (P)$ denote the mean of $P$, assumed finite. Consider the problem of testing whether the mean of the outcome changed after the enactment of the policy. Formally this problem is a test of the null hypothesis  $H_0 : \mu (P) = 0$ against the alternative $H_1 : \mu (P) \ne  0$.

It may be reasonable to assume that any changes to the outcome are symmetric around zero under the null hypothesis. That is, for the time being, we assume that distribution $P$ is symmetric about its median, e.g., the distribution is Gaussian or Uniform on $[-1,1]$.  In this case, under $H_0$, $P$ has mean 0 and is symmetric about 0. What transformations ${\bf G}$ apply?  Let  $( \epsilon_1 ,  \ldots , \epsilon_n )$ be any vector with each entry either $1$ or $-1$. Under $H_0$, the data $(X_1 , \ldots , X_n ) $ and $( \epsilon_1 X_1 , \ldots , \epsilon_n X_n )$ have the same distribution. The randomization hypothesis thus holds with ${\bf G}$ identified with the collection of $2^n$ vectors of length $n$ with entries $1$ or $-1$. A randomization test can be constructed by recomputing  any test statistic, such as the absolute sample mean $| \bar X_n | \equiv | \sum_{i=1}^n X_i  /n |$,  over the $2^n$ data sets of the form $( \epsilon_1 X_1 , \ldots , \epsilon_n X_n )$, as each $\epsilon_i$ varies in $\{ -1 , 1 \}$.

This example generalizes to  the multivariate or high-dimensional situation where   $X_i =  (X_{i,1} , \ldots , X_{i,d} )$ is vector in ${\bf R}^d$.  Testing a high-dimensional mean vector has received a lot of attention in recent years; see, e.g., \cite{wang2015high}, \cite{wang2019feasible} and \cite{huang2022overview}. The randomization hypothesis holds with the same group ${\bf G}$ if it is assumed $X_i$ and $- X_i$ have the same distribution under the null hypothesis. Thus, the construction will apply if one is willing to assume a semi-parametric model of symmetry under the null hypothesis.  Of course, one may wish to test a mean vector in a nonparametric setting that does not assume such symmetry. One of the themes of this paper is that randomization tests are often asymptotically valid, even when a suitable randomization hypothesis does not hold exactly.
\end{example}

\begin{example}[Two-Sample Permutation Tests]\label{example:2} Suppose that we observe the wages for two groups of workers, e.g., who face different labor regulations or belong to different demographic groups. Denote these measurements by $( X_1 , \ldots , X_m )$ and $(Y_1 , \ldots , Y_n )$, respectively. We might be interested in testing whether the two samples are drawn from the same distribution. 
Here, the null hypothesis $H_0$ is that the samples are generated from the same probability law. Under the null hypothesis $H_0$, the observations can be {\it permuted}, or assigned at random to either of the two groups, without changing the distribution of the data.  A little more formally, combine the data as  $ Z^{(N)} = (Z_1 , \ldots , Z_N ) = ( X_1 , \ldots , X_m, Y_1 , \ldots , Y_n )$.
Let $\pi = ( \pi (1) , \ldots , \pi (N) )$ be a permutation of $[N] = \{1, \ldots , N\}$, treated as an arbitrary bijection from $[N]$ to $[N]$. The permuted data $Z^{(N)}_{\pi}$ are given by
$
Z^{(N)}_{\pi} = ( Z_{\pi (1) } , \ldots , Z_{\pi (N)}  )~.
$
Under the null hypothesis $H_0$, the random variables  $Z^{(N)}$ and $Z^{(N)}_{\pi}$ have the same distribution. The randomization hypothesis thus holds with ${\bf G}$ identified with the collection of $N!$ permutations of the set $[N]$. Again, a randomization test can be constructed by recomputing any test statistic over the $N!$ data sets of the form $Z^{(N)}_{\pi}$. If the data are real-valued and the chosen test statistic is a function of the ranks of the observations, then the resulting test is a two-sample rank test. If the data have no ties, then the null distribution can be tabled. The classic example is the Wilcoxon (or, equivalently, the Mann-Whitney) test. Notice that nothing has required that the data be real-valued.  In fact, the data can be valued in any sample space. In particular, they may be vector-valued and high-dimensional, as is often the case in genomics. For some recent entries into this literature, see \cite{biswas2014nonparametric} and \cite{cousido2019two}.
\end{example}

\subsection{General Construction of a Randomization Test} \label{sec:generalconstruct}
Let $T(X)$ be {\it any}  real-valued test statistic chosen to test the null hypothesis $H_0$. Suppose the group ${\bf G}$ has $M$ elements. (In fact, the ideas generalize to infinite groups as well.)
Given $X = x$, let
$$T^{(1)} (x) \le T^{(2)} (x) \le \cdots \le T^{(M)} (x) $$
be the  values of $T(gx)$ as $g$ varies in ${\bf G}$, ordered from
smallest to largest.
Fix a nominal level $\alpha$, $0 < \alpha < 1$, and let $k$ be defined by
\begin{equation}\label{equation:13.rand.1}
k = M - \lfloor M \alpha \rfloor~ ,
\end{equation}
where $\lfloor M \alpha \rfloor$ denotes the largest integer less than or equal to
$M \alpha$.
Let $M^+ (x)$ and $M^0 (x)$ be the number of values $T^{(j)} (x)$, as $j$ varies over $1 , \ldots , M$,  which are greater than $T^{(k)} (x)$
and equal to $T^{(k)} (x)$, respectively.
Define the test function
$$a(x) = {{M \alpha - M^+ (x) } \over { M^0 (x) } }~ .$$
The assumption that ${\bf G} $ is a group ensures that the functions $T^{(k)} ( \cdot )$, $M^+ ( \cdot )$,
$M^0 ( \cdot )$, and $a ( \cdot )$ are invariant under ${\bf G}$; that is, $a (gx) = a (x)$ for all $g \in {\bf G}$ and $x \in {\cal X}$. 

Define  the {\it randomization test} function $\phi (X)$  to be equal to
1, $a(X)$, or 0 according to whether $T(X) > T^{(k)} (X)$, $T(X) = T^{(k)} (X)$, or $T(X) < T^{(k)} (X)$, respectively. Randomization tests have exact finite-sample Type 1 error control.
\begin{theorem}\label{theorem:1} \citep{hoeffding:1952} If $P \in \Omega_0$ satisfies the randomization hypothesis, then
\begin{equation}\label{equation:13.rand.3}
E_P [ \phi (X) ] = \alpha \text{ for all } P \in \Omega_0~.
\end{equation}
\end{theorem}
\begin{proof}
By construction, for every $x$ in ${\cal X}$,
\begin{equation}\label{equation:13.rand.2}
\sum_{g \in {\bf G}} \phi (gx) = M^+ (x) + a(x) M^0 (x) 
= M \alpha ~.
\end{equation}
Replacing $x$ by $X$ in  (\ref{equation:13.rand.2}) and taking expectations yields $$M \alpha = E_P \left[ \sum_{g \in {\bf G}} \phi ( gX ) \right] = \sum_{g \in {\bf G}} E_P [ \phi ( gX) ]~.$$ By the randomization  hypothesis, $E_P [ \phi ( gX ) ] = E_P [ \phi (X) ]$, so that $$M \alpha = \sum_{g \in {\bf G}} E_P [ \phi ( X) ] = M E_P [ \phi (X) ]~,$$ and the result follows.\hfill
\end{proof}

When $0 < a(x) < 1$, the above construction utilizes randomization in order to get exact level $\alpha$ Type 1 error.\footnote{Moreover, in some problems,  randomization tests form a complete class of tests, essentially meaning that one can restrict attention to the class of randomization tests without missing out on more powerful tests; see \cite{lehmann1949comments} or \cite{lehmann:romano:tsh:2022}, Corollary 5.11.1.  Put another way, any level $\alpha$ test may be replaced by a randomization test that is exact level $\alpha$ and with at least as much power against all alternatives.}  Indeed, the construction accounts for possible ties in the recomputed test statistics, as well as the possibility that $M \alpha$ is not an integer.  If one prefers to use a non-randomized test, then a conservative approach would be to reject the null hypothesis if and only if
$T(X)$ exceeds $T^{(k)} (X)$.  Put another way, the acceptance region
of this non-randomized test is
\begin{equation}\label{equation:accept}
\{ X : T( X) \le T^{(k)} (X) \}~.
\end{equation}
A corresponding conservative $p$-value (that avoids randomization in the vase $M \alpha$ is not an integer or the possibility of ties) is then
\begin{equation}\label{eq: p value conserv}
    \hat p = \frac{1}{M} \sum_{g \in {\bf G}} I \{ T(gX) \ge T(X) \}~~.
\end{equation}
That is, the $p$-value (\ref{eq: p value conserv}) satisfies $P\{\hat p \leq \alpha \}\leq \alpha$ for each $P\in\Omega_0$.

\begin{remark} \label{rem:samplingsubset}
When $|\mathbf G|$ is large, it is sometimes necessary to resort to a stochastic approximation to the randomization test described above.  With some care, this approximation can be done in a way that preserves the finite-sample validity of the randomization test.  In order to describe one such construction, let $g_1$ be the identity transformation, so $T(g_1X) = T(X)$, and let $g_2, \ldots, g_B$ be, e.g., i.i.d.\ $\sim \text{Unif}(\mathbf G)$.  It is possible to show that $\{ T(g_jX), 1 \leq j \leq B \}$ is exchangeable, from which it follows further that $$T(X) | T^{(1)}(X) \ldots, T^{(B)}(X) \sim \text{Unif}(\{T^{(1)}(X) \ldots, T^{(B)}(X)\})~,$$ where $T^{(1)}(X) \leq \cdots \leq T^{(B)}(X)$ are the ordered values of $T(g_jX), 1 \leq j \leq B$.  To see that $\{ T(g_jX), 1 \leq j \leq B \}$ is exchangeable,  the randomization hypothesis  implies that 
\begin{equation} \label{eq:exchG}
(T(g_1X), \ldots, T(g_BX)) \stackrel{d}{=} (T(g_1gX), \ldots, T(g_BgX))~,
\end{equation}
where $g \sim \text{Unif}(\mathbf G)$, independently of $g_2, \ldots, g_B$ and $X$.  But,  the $g_i g$  are i.i.d.\ $\sim \text{Unif}(\mathbf G)$, and hence $( g_1 g , \ldots , g_b g )$ is  exchangeable.    Therefore, the righthand-side of \eqref{eq:exchG} is exchangeable, and so the lefthand-side is as well.  It now follows that the construction above may be applied verbatim with $\mathbf G$ simply replaced by $\{g_1, \ldots, g_B\}$.
Moreover,  a valid $p$-value may be constructed as in (\ref{eq: p value conserv}) with $M$ replaced by $B$ and the sum is just over $g_1 , \ldots , g_B$. 
Finally, a similar argument holds if $g_2 , \ldots , g_B$ are sampled with replacement from ${\mathbf G}$ excluding the identity.  The only change is that $(gg_1 , \ldots , gg_B )$ can be shown to be exchangeable and distributed as $B$ elements taken randomly without replacement from ${\mathbf G}$. See \cite{hemerik2018exact} and \cite{ramdas2023permutation} for further discussion and alternatives.%
\end{remark}

\section{Approximation and Asymptotic Validity}

The randomization hypothesis may hold for the null hypothesis  $P \in {\Omega_0}$, while one is really interested in testing a different null hypothesis, say $P \in \bar \Omega_0$, where $\Omega_0 \subset \bar \Omega_0$. That is, in Example \ref{example:1}, one may be interested in testing whether each component of the mean of a random vector is equal to zero. Likewise, in Example \ref{example:2}, one may be interested in testing the whether the means of two distributions are equal. In many problems with this characteristic, randomization tests can still be quite useful. In particular, randomization tests are often asymptotically valid.

\subsection{Large-Sample Behavior of the Randomization Distribution}
Consider a sequence of situations with $X = X^{(n)}$,  $P = P_n$, ${\cal X} = {\cal X}_n$, ${\bf G} = {\bf G}_n$, $T = T_n$, etc.\  defined for each integer $n$. Let $\hat R_n$ denote the {\it randomization distribution} of $T_n$, defined by
\begin{equation}\label{equation:13.rand.5}
\hat R_n (t) = M_n^{-1} \sum_{g \in {\bf G}_n} I \{ T_n ( g X^{(n)} ) \le t \} ~.
\end{equation}
So, $\hat R_n$ is the distribution of $T_n ( G X^{(n)} )$ given $X^{(n)}$, where $G$ has the uniform distribution
on ${\bf G}_n$, independently of $^{(n)}$.\footnote{If the randomization hypothesis holds, then $\hat R_n$ is also the conditional distribution of $ T_n ( X^n )$ given the set $\{T^{(j)}(X^n) : 1 \leq j \leq M_n\}$.} When the group under consideration is based on permutations, the randomization distribution may also be called the permutation distribution. The threshold $T^{(k)}(X^{(n)})$ used in the construction of the randomization test $\phi $ in (\ref{equation:13.rand.2}) can be written
\begin{equation}\label{equation:juri} 
    \hat r_n (1- \alpha ) =  \inf \{  x: \hat R_n ( x ) \ge x \}~.
\end{equation}
Essentially, the randomization test rejects if $T_n (X^{(n)} ) > \hat r_n ( 1- \alpha )$, though the earlier construction accounts for the discreteness and possible ties in the randomization distribution. In most cases, the randomization distribution can be approximated by a continuous distribution and the need to randomize is asymptotically negligible. Thus, we often just use the slightly conservative test that rejects when $T_n(X^{(n)} ) > \hat r_n ( 1- \alpha )$.

In order to fully understand the operating characteristics of a randomization test,  particularly Type 1 error control and power, it is necessary to understand how the (random) critical value $\hat r_n (1- \alpha )$ and the randomization distribution $\hat R_n$ behave, both under the null hypothesis $H_0$ and under appropriate sequences of alternatives. Observe that 
\begin{equation}
E [ \hat R_n (t) ] = P \{ T_n ( G_n X^{(n)}  ) \le t \}~,
\end{equation}
where $G_n$ is, as before, a random variable distributed uniformly on ${\bf G}_n$, independently of $X^{(n)}$. If the randomization hypothesis holds, then 
\begin{equation}
E [ \hat R_n (t) ] = P \{ T_n ( X^{(n)} ) \le t \} = R_n (t)~,
\end{equation}
where $R_n$ is the true unconditional distribution of the test statistic $T_n$. So, if $T_n$ converges in distribution to a c.d.f.\ $R ( \cdot )$ which is continuous at $t$, it follows that 
\begin{equation}
E [ \hat R_n ( t) ] \to R(t)~.
\end{equation}
If, in addition, the randomization distribution $\hat R_n$ settles down to some nonrandom distribution, then it must be the same distribution as that of the unconditional limiting distribution of the test statistic $T_n$. 

\begin{remark}\label{remark:intuition} Intuitively, the randomization distribution should converge to the distribution $R$ in some asymptotic sense, at least under the randomization hypothesis. To see this, recall that, under the randomization hypothesis,  the  randomization test has exact level $\alpha$.  Apart from the issue of ties and discreteness in the randomization distribution, the (unconditional)  probability that $T_n$ exceeds $\hat r_n ( 1- \alpha )$ is close to $\alpha$, and so  $\hat r_n (  1- \alpha )$  should be close to the true quantile $r_n ( 1- \alpha )$.  But then, $r_n ( 1- \alpha )$  is close to the quantile $r ( 1- \alpha )$, at least if the limiting distribution $R ( \cdot )$ has a unique well-defined $1- \alpha$ quantile. This intuition is correct in most cases, but does not always hold, if, for example, the statistic $T_n (gX^{(n)})$ is constant as $g$ varies over $\mathbf{G}$.
\end{remark}

Hoeffding's condition reduces the study of the randomization distribution to that of a sequence of nonrandom distributions at the expense of the introduction of a little further randomness.  To this end, let $G_n'$ have the same distribution as $G_n$ above, with $X^{(n)}$, $G_n$, and $G_n'$ mutually independent.
\begin{definition}[Hoeffding's Condition]
We say that Hoeffding's condition holds if
\begin{equation}\label{equation:13.rand.6}
( T_n ( G_n X^{(n)} ) , T_n ( G_n' X^{(n)} )) \Larrow ( T , T' )
\end{equation}
under $P_n$, where $T$ and $T'$ are independent, each with common c.d.f.\ $R ( \cdot )$.
\end{definition}

\begin{theorem} Assume that Hoeffding's Condition holds and recall $\hat R_n ( \cdot )$ defined in (\ref{equation:13.rand.5}).  Then, under $P_n$,
\begin{equation}\label{equation:13.hatR}
 \hat R_n (t) \Parrow R(t)
\end{equation}
for every $t$ which is a continuity point of $R$. Let  $r ( 1- \alpha ) = \inf \{ t:~ R(t) \ge 1- \alpha \}.$ Suppose $R$ is continuous and strictly increasing at $r(1- \alpha )$. Then, under $P_n$, $$\hat r_n ( 1- \alpha ) \Parrow r( 1- \alpha )~.$$ Conversely, if (\ref{equation:13.hatR}) holds for some limiting  c.d.f.\ $R$ whenever $t$ is a continuity point, then (\ref{equation:13.rand.6}) holds.
\end{theorem}
\noindent
For a proof, see \cite{lehmann:romano:tsh:2022}, Theorem 17.2.3.

\subsection{The One-sample Problem\label{sec: one-sample}}

To fix concepts, let us return to Example \ref{example:1}. Let $P_n = P^n$ denote the joint distribution of the sample $X_1  , \ldots , X_n $. Assume that each $X_i$ has finite variance under $P$.  Here, characterizing the asymptotic behavior of the randomization distribution is particularly easy for the statistic $T_n = n^{1/2} \bar X_n$, whether or not the randomization hypothesis holds. There are several ways to see this.  First, note that the randomization distribution $\hat R_n$ is simply the distribution of $ n^{-1/2} \sum_i  \epsilon_i X_i$, conditioned on the data, where $\epsilon_1 , \ldots , \epsilon_n $ are i.i.d., each $1$ or $-1$ with probability $1/2$.  Clearly, this (conditional) distribution has mean zero and variance $n^{-1} \sum_{i=1}^n X_i^2$.  By an appropriate central limit theorem, this distribution converges to a normal distribution with mean zero and variance $E_P[X_i^2]$ (not $\text{Var}_P(X_i)$). Indeed, Hoeffding's Condition can be verified  with $R$ given by $N(0,E_P[X_i^2])$; see \cite{lehmann:romano:tsh:2022}, Example 17.2.4.
 
Although this calculation was straightforward, it is perhaps useful to provide some additional intuition, which will generalize to other choices of $T_n$.  The behavior of the randomization distribution when $X_i$ has distribution $P$ is the same as the behavior of the randomization distribution when $X_i$ has distribution $P^s$, where $P^s$ is the symmetrized version of $P$, i.e., $P^s$ is the unconditional distribution of $\epsilon_i  |X_i |$.   To see this, observe that the randomization distribution depends only on the values $|X_1| , \ldots , |X_n |$.  But the distribution of $|X_i|$ under $P$ is the same as that under $P^s$. So, the problem of studying the behavior of $\hat R_n$ under $P$ is reduced to the equivalent problem under $P^s$. But the randomization hypothesis holds under $P^s$ and so the behavior of  $\hat R_n$  under $P^s$  should be that of the limiting unconditional distribution of the test statistic. By the ordinary central limit theorem, this limiting distribution is $N(0, \text{Var}_{P^s}[X_i] ) = N(0, E_P[X_i^2])$. Therefore, we can conclude that, under any $P$ with finite variance $\sigma^2 (P)$, 
$$\hat R_n ( t ) \Parrow \Phi ( t / \sigma (P)) \quad\text{and}\quad \hat r_n ( 1- \alpha ) \Parrow \sigma (P) z_{1- \alpha }~.$$
Let $\phi_n$ be the level $\alpha$ randomization test  based on the test statistic $T_n$.  If the distribution $P$ has mean 0, then it easily follows that $E_P[\phi_n] \to \alpha$.

To summarize,  for the problem of testing whether the mean of $P$ is zero against the alternative that the mean exceeds zero (or, with similar results, against two-sided alternatives), the randomization test is asymptotically level $\alpha$. Of course,  $\phi_n$ is {\it exact} level $\alpha$ if the underlying distribution is symmetric about zero; otherwise, it is at least asymptotically pointwise level $\alpha$, as long as $\sigma^2(P) < \infty$.

One can also easily derive the (local) limiting power of the randomization test. Assume that the underlying distribution $P_n$ is $N(hn^{-1/2}, \sigma^2)$.  First, by the above reasoning (or, by verifying Hoeffding's Condition), if $h = 0$, then $\hat r_n ( 1- \alpha ) \to \sigma z_{1- \alpha}$ in probability. By contiguity, it follows that, under $N( hn^{-1/2} , \sigma^2 )$, $\hat r_n ( 1- \alpha ) \to \sigma z_{1- \alpha}$ in probability as well. Under $N( hn^{-1/2} , \sigma^2  )$, the test statistic $T_n$  converges in distribution to  $N( h , \sigma^2 )$. Therefore, by Slutsky's Theorem, the limiting power of the test $\phi_n$ against the alternative $N( hn^{-1/2} , \sigma^2 )$ is given by
\begin{equation}
    E_{P_n} [\phi_n] \to P \{ \sigma Z+h > \sigma z_{1- \alpha} \} = 1 - \Phi  \left ( z_{1- \alpha} -{h \over {\sigma}} \right  )~.
\end{equation}
In fact, this is also the limiting local power of the uniformly most powerful (UMP) test when the variance is known; it is also the local limiting power of the uniformly most powerful unbiased (UMPU) test when the variance is unknown. Hence, to first order,  there is asymptotically no loss in power when using the randomization test, as opposed to the  UMP test, but the randomization test  has the advantage that its size is $\alpha$ over all symmetric distributions.  In statistical terminology, we say that the asymptotic relative efficiency of the randomization test with respect to the UMP or UMPU test is one. In fact, the relative efficiency is one whenever the underlying family is a quadratic-mean-differentiable location family with finite variance.

One  benefit of randomization tests is that one does not have to assume a parametric model, like normality.  In practice, critical values can be obtained from the exact randomization distribution, or its Monte Carlo approximation obtained by randomly sampling elements of ${\bf G}$.  In summary, two additional benefits are revealed by asymptotics. First, the randomization test may be used in large samples even when the randomization hypothesis fails; in the one-sample case, this means the assumption of symmetry is not required.  Second, asymptotics allow us to perform local power calculations and show that, even under normality, very little power is lost when using a randomization test as compared to the UMP or UMPU test; in fact, the randomization test and the UMP or UMPU test have the same limiting local power function against normal contiguous alternatives.

The above considerations generalize in two important ways: the arguments apply to other test statistics and to vector-valued observations. For example, consider a test statistic $T_n$ with some distribution $P$ that is symmetric about 0. Assume that $T_n$ is asymptotically linear in the sense that, for some  (influence) function $\psi_P ( \cdot )$, assumed to be an odd function, we can write
\begin{equation}\label{equation:sign1}
T_n = n^{-1/2} \sum_{i=1}^n \psi_P ( X_i ) + o_P (1)~,
\end{equation}
where  $E_P [ \psi_P (X_i ) ] = 0$ and $\tau_P^2 = \text{Var}_P [ \psi_P ( X_i ) ] < \infty$. Let $\hat R_n$ denote the randomization distribution based on $T_n$ and the group of sign changes (applied to the $X_i$).  Then, Hoeffding's condition holds with $P_n = P^n$ and $R(t) = \Phi ( t / \tau (P))$. As a consequence
\begin{equation}
    \hat R_n ( t) \Parrow \Phi ( t / \tau (P))~.
\end{equation}
But, by the argument presented earlier, the  behavior of the randomization distribution under an asymmetric distribution $P$ is the same as that under the symmetrized distribution $P^s$. Thus, we can also conclude $$\hat R_n (t) \Parrow \Phi ( t / \tau (  P^s ))  \text{ and }
\hat r_n ( 1- \alpha ) \Parrow \tau (  P^s) z_{1- \alpha}~.$$ Such results can be applied to general location models; see \cite{lehmann:romano:tsh:2022}, Theorem 17.2.4.   Under general conditions, the limiting local power of the randomization test based on an optimal test statistic (such as Rao's score statistic) is the same as that of the Rao\footnote{See Section 14.4.3 of \cite{lehmann:romano:tsh:2022} for a discussion of Rao's score test.} test and hence is locally asymptotically uniformly most powerful.   The randomization test, however, is robust to model misspecification. As an example, the sample median is an optimal estimator in a double exponential location model. If a randomization test is based on the sample median, it exactly controls Type 1 error for symmetric distributions, asymptotically controls Type 1 error under asymmetry, and has optimal local asymptotic power under the assumed model.

Note that \cite{hartigan1969using} made effective use of the group of sign changes in order to construct what he called typical values for a location parameter, though he did not seem to realize the connection with \cite{hoeffding:1952}.
As a side note,  the use of sign changes is related to the wild bootstrap 
of \cite{wu1986jackknife}; also see \cite{liu1988bootstrap}.

\subsection{Testing Randomness and the Hot Hand Fallacy\label{sec:hothand}}

The features of the simple example discussed in Section \ref{sec: one-sample} play out in real data. In a landmark paper in behavioral economics, \cite{tversky1971belief} posit that people tend to find small samples ``overly representative'' of the populations from which they are drawn. They refer to this tendency as belief in the ``law of small numbers.'' One of the main pieces of empirical support for the law of small numbers comes from the literature on the ``Hot Hand Fallacy,'' initiated by \cite{gilovich1985hot}, henceforth GVT. GVT document that there is widespread belief in the ``Hot Hand'' in basketball. That is, people believe that basketball players are more likely to make a shot after making several shots than after missing several shots. GVT hypothesize that this belief is erroneous, i.e., that sequences of basketball shots are i.i.d., asserting that people infer positive dependence from randomly occurring streaks of consecutive makes or misses. 

To test this hypothesis, GVT conduct a controlled experiment. They arrange for the members of the Cornell University men's and women's varsity basketball teams to each shoot $n=100$ consecutive shots, and record whether each shot was made or missed. For each shooter $i$, let the binary variable $X_{i,j}$ indicate whether the $j$th shot was made or missed. Formally, GVT are interested in testing the null hypotheses
\begin{equation}\label{eq: individual null}
    H_{0}^{i}:\text{The sequence }X_{i,1}\ldots,X_{i,n}\text{ is i.i.d.\ Bernoulli with unknown success rate $q_i$}~,
\end{equation}
for each shooter $i$, against alternatives in which the probabilities of makes and misses immediately following streaks of consecutive makes or consecutive misses are greater than their unconditional probabilities. To do this, they consider test statistics of the form
\begin{equation}\label{eq: D k def}
    D_{i,k} = \frac{1}{\vert \mathsf{Make}_k \vert} \sum_{j \in \mathsf{Make}_k} X_{i,j} - \frac{1}{\vert \mathsf{Miss}_k \vert} \sum_{j \in \mathsf{Miss}_k} X_{i,j}~,
\end{equation}
where $\mathsf{Make}_k$ and $\mathsf{Miss}_k$ are the sets of indices $j$ in $1,...,n$ such that $X_{i,j-1}=1,\ldots,X_{i,j-k}=1$ and $X_{i,j-1}=0,\ldots,X_{i,j-k}=0$, respectively. That is, the statistic $D_{i,k}$ measures the difference between the proportion of makes following $k$ consecutive makes and $k$ consecutive misses. GVT compare measurements of $D_{i,k}$ to critical values obtained by viewing $D_{i,k}$ as the test statistic in a two-sample t-test. They find that they cannot reject the null hypotheses \eqref{eq: individual null}. This finding became widely cited and used as support for economic and behavioral models that incorporate belief in the law of small numbers \citep{thaler2009nudge,barberis2003survey}.

\cite{miller2018surprised} make two interesting observations, raising doubt in the GVT result. First, they observe that, under the null hypothesis (\ref{eq: individual null}), the statistic $D_{i,k}$ is negatively biased.\footnote{Roughly, conditioning on directly following an observed streak of makes or misses introduces a mechanical bias in the probability of observing a make, that can be large in small samples. This bias is related to the bias in dynamic fixed-effect models documented in \cite{nickell1981biases}. See Appendix D.1 of \cite{miller2018surprised} for further discussion.} That is, the expectation of the statistic $D_{i,k}$ is less than zero under the null hypothesis \eqref{eq: individual null}. Second, they observe that, under the null hypothesis \eqref{eq: individual null}, the distribution of the sequence $X_{i,1}\ldots,X_{i,n}$ is invariant to permutations. Consequently, tests with exact finite-sample error control can be constructed by computing the randomization distribution associated with the statistic $D_{i,k}$. They argue that the randomization tests reverse the GVT results, i.e, that the negative bias was masking evidence of positive sequential dependence.

To get a quantitative sense of the factors at play here, Figure \ref{fig: hot-hand} displays histograms of the randomization distribution for the statistic $D_{i,k}$ for two shooters from the GVT experiment, with $k = 3$. The observed values of each statistic are displayed with vertical dashed teal lines. The means of the randomization distribution are displayed with vertical dotted black lines. These means are meaningfully different from zero, and can be treated as estimates of the bias of $D_{i,k}$ under the the null hypothesis (\ref{eq: individual null}). The randomization $p$-value \eqref{eq: p value conserv} is given by the total mass of the randomization distribution that exceeds the observed statistics. The randomization $p$-values for the two shooters displayed in Figure \ref{fig: hot-hand} are $0.0008$ and $0.375$, respectively. 

\begin{figure}[t]
\begin{centering}
\caption{Normal Approximation to the Randomization Distribution}
\label{fig: hot-hand}
\medskip{}
\begin{tabular}{c}
\includegraphics[scale=0.4]{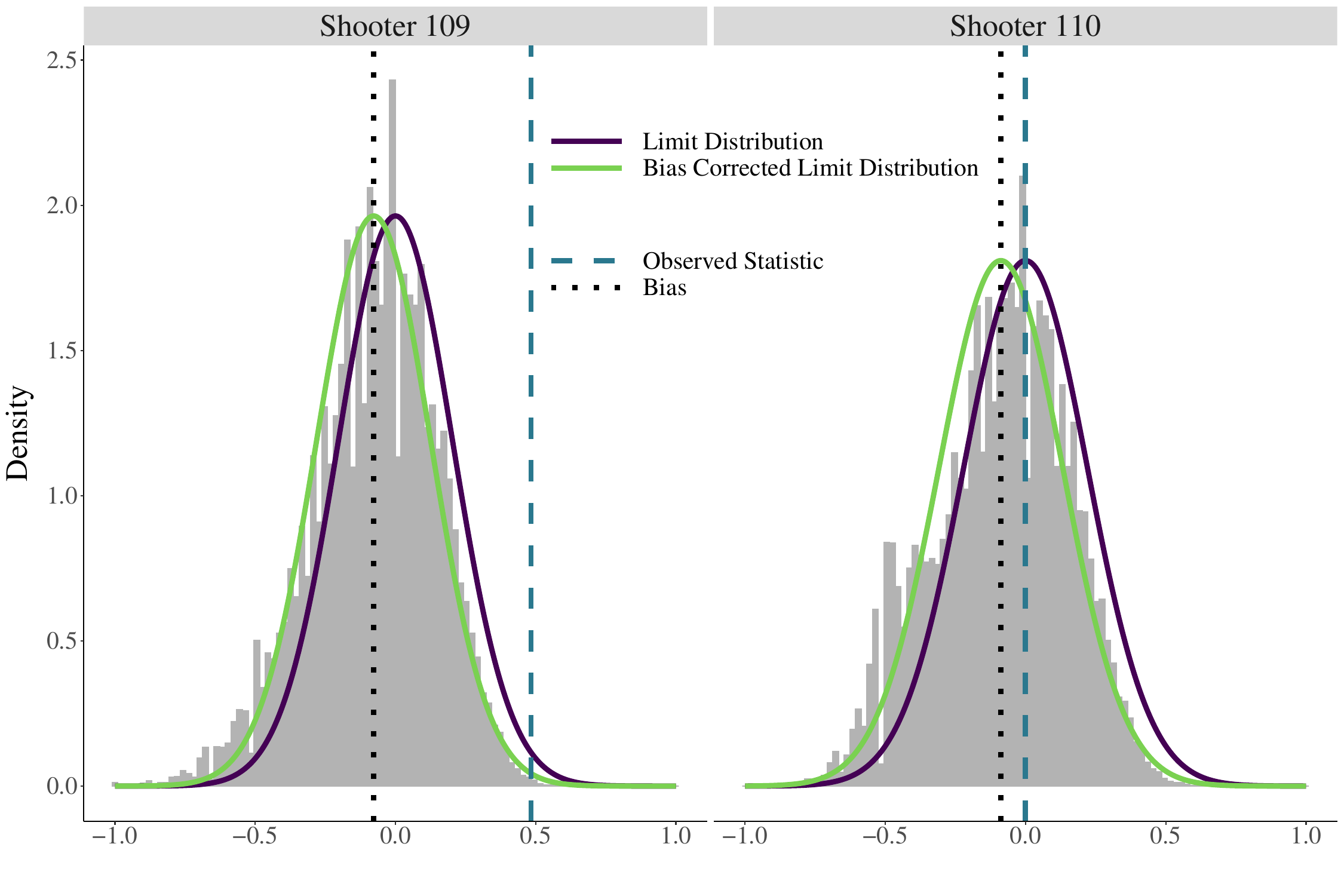}\tabularnewline
\end{tabular}
\par\end{centering}
\medskip{}
\justifying
{\footnotesize{}Notes: Figure \ref{fig: hot-hand} displays histograms of the randomization distributions associated with the statistic $D_{i,k}$ for two shooters from the GVT experiment. In both cases, we take $k=3$. The observed values of the statistic $D_{i,k}$ are displayed with teal vertical dashed lines. The mean values of the randomization distributions are displayed with a black vertical dotted line. The purple curves give the densities of the limiting normal distribution $N\left(0,\sigma_k^{2}(\hat{q}_i)\right)$, scaled by $1/\sqrt{n}$, where we recall that the observed success probability is given by $\hat{q}_i = n^{-1}\sum_{j=1}^n X_{i,j}$. The green curves are leftward shifts of the purple curve by the means of the randomization distributions.}{\footnotesize\par}
\end{figure}

\cite{ritzwoller2022uncertainty} revisit these data and give the following result.
\begin{theorem}\label{thm: hot hand theorem}
Let $\hat{R}_{n,k}(t)$ denote the randomization distribution associated with the statistic $D_{i,k}$. If the sequence $X_{1,i},\ldots,X_{n,i}$ is i.i.d.\ $\mathsf{Bernoulli}(q)$ random variables, then
\begin{equation}
\sqrt{\frac{n}{\sigma^2_k(q)}}\hat{D}_{i,k}\overset{d}{\rightarrow}\mathcal{N}\left(0,1\right)
\quad\text{and}\quad
\hat{R}_{n,k}(t) \Parrow \Phi ( t / \sigma_k(q) )~,\label{eq: hot hand convergence}
\end{equation}
where $\sigma_k^{2}(q)=\left(q\left(1-q\right)\right)^{1-k}(\left(1-q\right)^{k}+q^{k})$.
\end{theorem}

\noindent The convergence \eqref{eq: hot hand convergence} verifies that the randomization distribution of the test statistic $\hat{D}_{i,k}$ settles down to its limiting normal distribution. Even for $n=100$, the limiting normal distribution is a good approximation to the randomization distribution. The purple curves in Figure \ref{fig: hot-hand} display densities of the normal distributions $N\left(0,n^{-1}\sigma_k^{2}(\hat{q}_i)\right)$, where we proxy the unknown true success probabilities $q_i$ with their plug-in estimates $\hat{q}_i = n^{-1}\sum_{j=1}^n X_{i,j}$. The green curves result from shifting the purple curves to the left by the means of the randomization distributions, i.e., they display a bias-correction to the limiting normal approximation.

Moreover, \cite{ritzwoller2022uncertainty} demonstrate that results analogous to Theorem \ref{thm: hot hand theorem} hold even if the null hypothesis \eqref{eq: individual null} is violated. In particular, they give simple regularity conditions under which statements analogous to \eqref{eq: hot hand convergence} hold, that permit the distribution of the sequence $X_{1,i},\ldots,X_{n,i}$ to exhibit positive sequential dependence. These results allow \cite{ritzwoller2022uncertainty} to measure the power of the randomization test exhibited in Figure \ref{fig: hot-hand} against realistic alternatives. They find that GVT did not collect enough data to detect reasonable departures from randomness. Moreover, they conclude that the \cite{miller2018surprised} finding that the data deviate from randomness was entirely driven by Shooter 109, displayed on the left in Figure \ref{fig: hot-hand}. If this shooter is omitted from the sample, the data are insufficient to reject the null hypotheses (\ref{eq: individual null}). That is, the empirical literature of the hot hand fallacy is founded on conclusions drawn from insufficient samples. In our view, this is itself compelling, if circumstantial, evidence of belief in the law of small numbers. 

\section{Two-sample Permutation Tests}\label{section:two}
	
Consider again the setting of Example \ref{example:2}. That is, assume   $X_1, \ldots, X_m$ are i.i.d.\ $P$ and, independently, $Y_1, \ldots, Y_n$ are i.i.d.\ $Q$.  For example, one group may be considered a treatment group, and the other a control group. In many cases, researchers may be interested in testing whether $P=Q$, or perhaps just $\theta (P) = \theta (Q)$, where $\theta (P)$ could be the mean of $P$, a quantile of $P$, or any other parameter or functional. See \cite{bertanha2023permutation} for a treatment of two-sample permutation tests of differences between functionals estimable at nonparametric rates.

\subsection{Large-Sample Behavior}

For testing equality of distributions $P=Q$, permutation tests achieve exact Type 1 error control, as the randomization hypothesis holds. For example, one could base a test on the two-sample Kolmogorov-Smirnov statistic, or some generalization, such as comparing empirical probabilities over a Vapnik-Cervonenkis class, as in \cite{romano:1990}. Such an omnibus test would have exact Type 1 error control, and would be consistent in power against any distributions $P$ and $Q$ with $P \ne Q$. For testing equality of means, the difference in sample means would be a more appropriate test statistic. In this case, as we will see, permutation tests may still fail to control the Type 1 error rate, even asymptotically. 

Assume that estimators $\hat \theta_m$ and $\hat \theta_n$ of $\theta ( P)$ and $\theta (Q)$ are asymptotically linear in the sense that
\[m^{1/2} ( \hat \theta_m  - \theta (P) ) = \frac{1}{\sqrt{m}} \sum_{i=1}^m \psi_P ( X_i ) + o_P(1)~, \]
for some function $\psi_P(\cdot)$, where $\E_P \psi_P(X_i) = 0$ and $0< \sigma^2 (P) = \text{Var}_P (\psi_P (X_i )) < \infty$. Consider the test statistic\footnote{Note that the factor $m^{1/2}$ plays no role, as the permutation test with or without such a constant factor results in the same $p$-value; it is only used to ensure that the permutation distribution has a non-degenerate limit distribution.}
\begin{equation}\label{eq: diff in theta}
    T_{m,n} = m^{1/2} ( \hat \theta_m ( X_1 , \ldots , X_m ) - \hat  \theta_n ( Y_1 , \ldots , Y_n ) )~.
\end{equation}
In general, the permutation distribution fails to recover the true null sampling distribution, as seen in the following theorem, stated as Theorem 2.1 in \cite{chung2013exact}. The proof follows by verifying Hoeffding's Condition.
		
\begin{theorem}\label{theorem:two}
Assume $X_1 , \ldots , X_m$ are i.i.d.\ $P$ and, independently, $Y_1 , \ldots , Y_n$ are i.i.d.\ $Q$.   Let $m \to \infty$, $n \to \infty$, with
	$N = m+n$, $p_m = m/N$  and $p_m \to p \in (0,1)$ with $ p_m - p = O( m^{-1/2} )~.$
	Let $\bar P$ be the mixture distribution: $\hat P = pP + (1-p) Q$.
Consider testing the null hypothesis $H_0: \theta(P) = \theta(Q)$ based on a test statistic \eqref{eq: diff in theta} where the estimator $\hat \theta_n$ is assumed to be asymptotically linear under $P$, $Q$ and $\bar P$. Then, the permutation distribution $\hat R_{m,n} ( \cdot )$ based on $T_{m,n}$ satisfies
\begin{equation}
    \sup_t | \hat R_{m,n}  (t) - \Phi ( t / \tau (\bar P))| \pto 0~,
\end{equation}
where $\tau^2 ( \bar P ) = \frac{1}{1-p} \sigma^2 ( \bar P )$.
\end{theorem}

Note that, for the case of the mean (or linear functionals),  $\tau^2 (\bar P) = p (1-p)^{-1} \sigma^2(P) + \sigma^2(Q)$.
By the central limit theorem, the true unconditional sampling distribution of $T_{m,n}$ satisfies 
\begin{equation}
T_{m,n} \approx \mathcal{N}\left(0, \sigma^2(P) + p(1-p)^{-1}\sigma^2(Q)\right)~.\label{eq: dem sample}
\end{equation}
By contrast, Theorem \ref{theorem:two} demonstrates that the permutation distribution satisfies
\begin{equation}
\hat{R}_{m,n} \approx \mathcal{N} \left(0, p (1-p)^{-1} \sigma^2(P) + \sigma^2(Q) \right )~. \label{eq: dem perm}
\end{equation}
Thus, the permutation distribution is asymptotically different than the true sampling distribution, unless asymptotic variances are equal. 
In the case of the mean, equality occurs if and only if $p = 1/2$ or the variances of $P$ and $Q$ agree. Consequently, the critical value from the permutation distribution may be inconsistent and permutation tests based on $T_{m,n}$ may yield a rejection probability arbitrarily close to 1 for two-sided testing (and close to 1/2 for one-sided testing).
For example, if the true limiting variance is very large but the variance of the permutation distribution is very nearly zero (so that the critical value is nearly 0), then the chance that permutation test rejects can  be quite large.
More specifically, take $P$ to be $N(0, B )$ and $Q$ to be $N( 0,  \epsilon )$.
Then, the true limiting distribution has variance $B + \frac{p \epsilon }{1-p}$,
which can be made arbitrarily large by taking $B$  large.  On the other hand, the limiting variance of the permutation distribution is $\epsilon + \frac{Bp}{1-p}$,
which can be made arbitrarily small by choosing $\epsilon$ and $p$ small enough.
			
Next, consider testing equality of medians based on the difference of sample medians.  If both groups have median $\theta$, the true asymptotic variance  of the normalized difference is 
$$\frac{1}{4 f_P^2 ( \theta )} + \frac{p}{1-p} \frac{1}{4 f_Q^2 ( \theta ) } ~,$$
where $f_P ( \theta )$ and $f_Q ( \theta )$ denote the densities of the distributions $P$ and $Q$ evaluated at $\theta$. By contrast, the asymptotic variance of the permutation distribution is $$ \frac{1}{1-p} \sigma^2 ( \bar P ) = \frac{1}{1-p} \frac{1}{4 ( pf_P ( \theta )  + (1- p) f_Q ( \theta))^2}~.$$ These match  if and only  if $f_P ( \theta ) = f_Q ( \theta )$, which clearly need not hold.
		
\subsection{Inconsistent Error Control}

The incongruity between the unconditional sampling distribution \eqref{eq: dem sample} and the permutation distribution \eqref{eq: dem perm} has a severe adverse effect on error control. To see this, consider testing the null hypothesis $H_0$ of  equality of means against two-sided alternatives, based on the test statistic $ | \bar X_m - \bar Y_n |$. Suppose that the distribution $Q$ has a much smaller variance than the distribution $P$ and the proportion of the total observations sampled from $Q$, i.e., $p_m$ is small. In this case, the variance of the unconditional sampling distribution \eqref{eq: dem sample} is larger than the variance of the permutation distribution \eqref{eq: dem perm}, and so the permutation test will over reject, under the null. As the ratio of the variances of $Q$ and $P$ converges to zero, the rejection rate converges to one. 

Similarly, suppose that the mean of $Q$, $\theta(Q)$, is slightly bigger than the mean of $P$, $\theta(P)$. In this case, the probability of the event that the null hypothesis $H_0$ is rejected and $\bar X_m > \bar Y_n$ is nearly one half, where we recall that $X_i$ and $Y_i$ are samples from $P$ and $Q$, respectively. That is, the null hypothesis $H_0$ is rejected on the basis of evidence $\bar X_m > \bar Y_n$ contradictory to the truth $\theta(Q) > \theta(P)$. This is referred to as a directional, or Type III, error \citep{mosteller1948k}.\footnote{Errors of this form are dangerous, because, inevitably, a rejection of $H_0$  would be accompanied by a statement that the difference in means is positive. Indeed, one can view 
directional errors as even more serious than Type 1 errors. \cite{tukey1991philosophy} writes
``Statisticians classically asked the wrong question--and were willing to answer with a lie, one that
was often a downright lie.  They asked `Are the effects of $A$ and $B$ different?' and they were willing
to answer `no'.
All we know about the world teaches us that the effects of $A$ and $B$ are always different --
in some decimal place -- for any $A$ and $B$.  Thus asking `Are the effects different?' is foolish.
What we should be answering first is `Can we tell the direction in which the effects of $A$ differ
from the effects of $B$?." Technically, lack of Type 1 error control implies lack of Type 3 error control, and it is important to control both.  Similarly, low power, or high Type 2 error, can result without proper choice of test statistic.}

As will be seen shortly, the problem of mismatched asymptotic variances can be fixed.  But first, we provide some intuition. Recall that when $P=Q$, the permutation distribution should reflect the true sampling distribution, as mentioned in Remark \ref{remark:intuition}. So, asymptotically, the permutation distribution should  settle down to  the true unconditional distribution of the statistic $T_{m,n}$, at least when $P = Q$. Now, the permutation distribution is invariant with respect to ordering of the data. Thus, even when $P \ne Q$, the data  should behave similarly to the situation when all $N = m+n$ observations are  i.i.d. from the mixture distribution $\bar P = pP + (1-p)Q$, where $p =\frac{m}{N}$.   Therefore, the behavior of the permutation distribution under $(P,Q)$ should be approximately the same as under $( \bar P, \bar P )$.\footnote{This is not a formal argument. The result does hold under the assumptions in Theorem \ref{theorem:two}, but must be verified by technical arguments.}

Let $J_{m,n} ( P, Q )$ be the distribution of the test statistic sequence $T_{m,n}$ when $m$ i.i.d.\ observations are sampled from $P$ and $n$ i.i.d.\ observations are samples from $Q$. By the argument above, the permutation distribution $\hat R_{n,m}$ under $P = Q$  should satisfy
\begin{equation}\label{equation:pivot}
\hat R_{m,n} \approx J_{m,n} ( \bar P , \bar P ) \dto J ( \bar P , \bar P ) = J (P,Q)~,
\end{equation}
where $J (P, Q)$ is the limiting distribution of $T_{m,n}$ under $(P, Q)$ and may depend on $p = \lim m/N$. When $J( P, Q)$ does not equal $J ( \bar P , \bar P )$, the permutation distribution will be inconsistent. 

This suggests a solution. As the problem occurs when $J ( P, Q) \ne J ( \bar P , \bar P )$, consistency should result if an asymptotically pivotal test statistic is chosen. That is, the problem is resolved if $J (  P , Q)$ is the same for all $(P , Q)$. To summarize,
if, when the null hypothesis $H_0$ is true, the true sampling distribution $J_{m,n}(P,Q)$ of the test statistic $T_{m,n}$ under $(P,Q)$ is 
asymptotically pivotal, then
\begin{equation}
    J_{m,n}  (P,Q)  \approx J_{m,n} ( \bar P , \bar P ) \approx\hat R_{m,n, }
\end{equation}
when sampling from either $( \bar P , \bar P )$ or $(P , \bar Q )$.

\subsection{Studentization}

A simple way to achieve asymptotic pivotality is to appropriately studentize the test statistic. Suppose that $\hat \sigma^2_m ( X_1 , \ldots , X_m )$ is estimator of the variance $\sigma^2(P) = \Var_P (\psi_P(X_i))$. Let
\begin{equation} \label{eq: two sample student stat}
    S_{m,n}(Z^{(N)}) = \frac{T_{m,n}}{\sqrt{ \frac{N}{m}\sigma^2_m(X_1,\ldots,X_m)+ \frac{N}{n} \sigma^2_n(Y_1,\ldots,Y_n)}}
\end{equation}
denote the studentized test statistic. The permutation distribution of the statistic $S_{m,n}$ is 
\begin{equation}\label{eq: two sample student dist}
\hat{R}^S_{m,n}(t) = \frac{1}{N!}\sum_{\pi \in G_N} \I\{S_{m,n}(Z_{\pi(1)}, \ldots, Z_{\pi(N)}) \leq t \} ~,
\end{equation}
where $\mathbf{G}_{N}$ denotes the $N!$ permutations of $\{1, \ldots, N\}.$ 

The studentized permutation test is just as easy to implement as an unstudentized test. Its asymptotic validity is established by the following theorem, stated as Theorem 2.2 in \cite{chung2013exact}.
\begin{theorem}
Assume the setup and conditions of Theorem \ref{theorem:two}.   Suppose that $\hat \sigma_m ( X_1 , \ldots , X_m )$ is a consistent estimator of $\sigma (P)$ when $X_1 , \ldots , X_m$
are i.i.d. $P$.  Assume consistency also under $Q$ and $\bar P$. The permutation distribution \eqref{eq: two sample student dist} based on the statistic \eqref{eq: two sample student stat} satisfies
\[ \sup_t | \hat R^S_{m,n} (t) - \Phi (t) | \pto 0~.\]
\end{theorem}
\noindent To summarize, asymptotically, the permutation distribution of the studentized statistic now behaves like the true unconditional distribution  of the studentized statistic. The permutation test based on a studentized statistic is asymptotically valid for testing $\theta (P) = \theta (Q)$ and retains exact Type 1 error control when $P = Q$. The technical conditions are very weak and do not require strong differentiability assumptions.  Under such conditions, the bootstrap may not even be first order asymptotically correct. See \cite{yadlowsky2021evaluating} and \cite{ritzwoller2024uniform} for further applications that leverage this generality.

\begin{figure}[t]
\begin{centering}
\caption{Stable Error Control with Studentization}
\label{fig: cps}
\medskip{}
\begin{tabular}{c}
\includegraphics[scale=0.4]{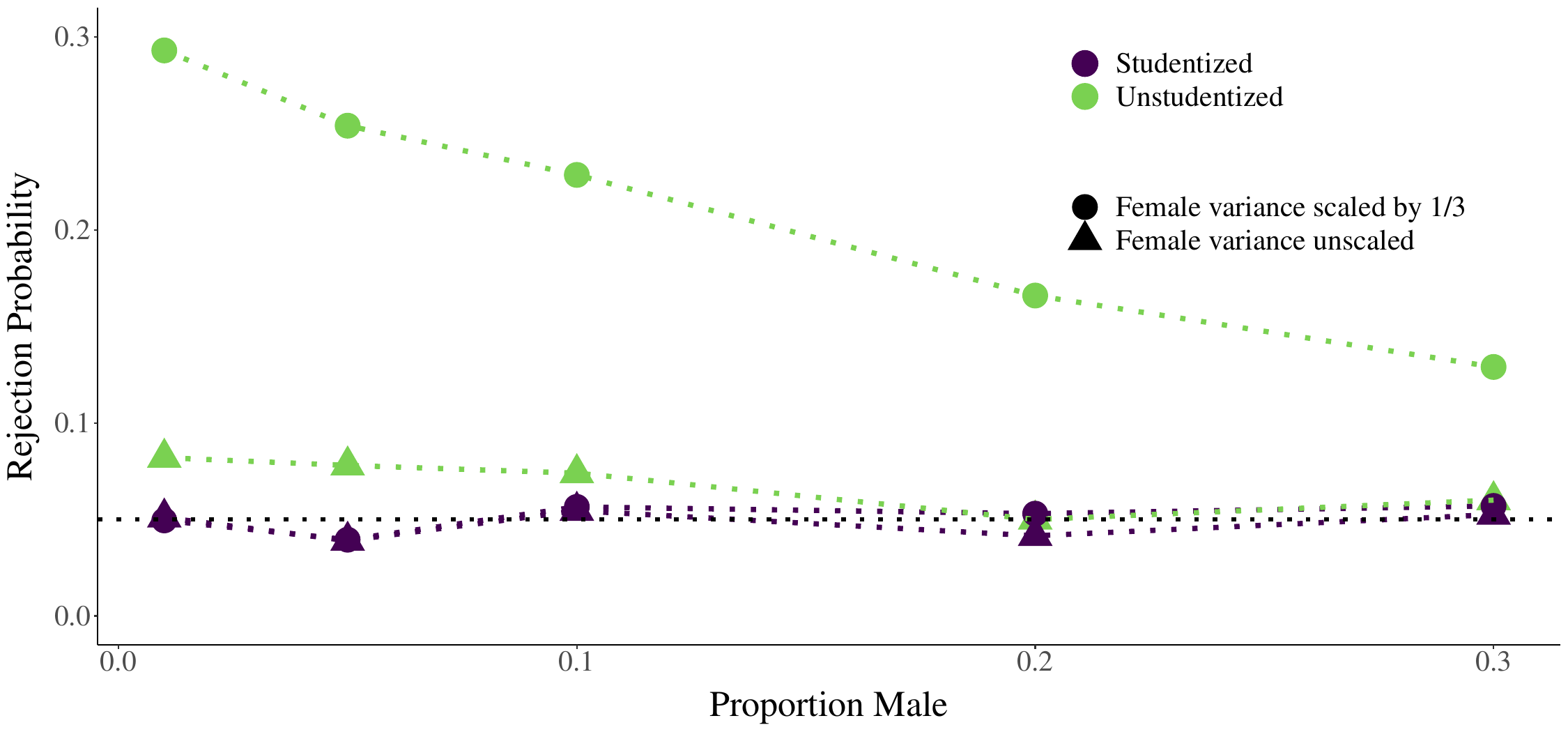}\tabularnewline
\end{tabular}
\par\end{centering}
\medskip{}
\justifying
{\footnotesize{}Notes: Figure \ref{fig: cps} displays the results of a Monte Carlo experiment. The experiment is implemented using measurements of the log earnings for populations of men and women from the March 2009 Current Population Survey (CPS). See the main text for further information on the design of the experiment. Each measurement is made using 2,000 samples of $N=100$ observations from the CPS, where we vary the proportion $p$ of men in the sample. Randomization tests are implemented using 10,000 permutations.}{\footnotesize\par}
\end{figure}

To get a quantitative view of the difference between randomization tests implemented with and without studentized test statistics, we conduct a  Monte Carlo experiment. We obtain measurements of the log earnings of 50,742 fully employed individuals from the March 2009 Current Population Survey (CPS).\footnote{These data were downloaded from Bruce Hansen's website: \url{https://users.ssc.wisc.edu/~bhansen/econometrics/}.} We consider two settings. In both settings, for the purposes of the experiment, we normalize the measured log earnings from 29,140 men and 21,602 women so that both have means equal to zero.\footnote{The average measured log earnings for men and women in the sample are 10.8 and 10.5, respectively.} In the first setting, we use these normalized distributions as the population distributions $P$ and $Q$, respectively. In this case, the variance of log earnings for men (0.531) is larger than for women (0.392). In the second setting, we reduce the variance of the log earnings for the female population by a factor of 1/3, i.e., we multiply these data by $1/\sqrt{3}$

In the experiment, we consider a random sample of $N=100$ measurements of log earnings from the CPS data. Here, $m = p\cdot N$ of the observations are from men. We are interested in testing the null hypothesis that the average log earnings of men and women are equal. As we have normalized the two populations $P$ and $Q$ to both have means equal to zero, the null hypothesis is true. Figure \ref{fig: cps} displays measurements of the rejection rates of the randomization tests using the studentized and unstudentized difference in means, where we vary the proportion $p$ of men in the sample. In both settings, the randomization test that uses the unstudentized test statistic over-rejects. This over-rejection is more severe in the setting where the variance of the measurements of female log earnings has been scaled to be smaller than the variance of male log earnings. On the other hand, in both settings, the rejection rate for the randomization test using the studentized statistic is approximately correct. 

\subsection{Generalizations}

The intuition for the results for the two-sample problem apply quite generally.  We will touch upon four possibilities: (i) higher-order kernel $U$-statistics, (ii) multiple samples, (iii) multivariate observations, and (iv) multiple testing.

\subsubsection{$U$-statistics}  Rather than basing a two-sample test on a statistic that is a difference of estimators, one can consider
the class of two-sample $U$-statistics of the form:
\begin{equation}
    U_{m,n}(Z)  = \frac{1}{{m \choose r}{n \choose r}} \sum_{\alpha} \sum_{\beta} \varphi(X_{\alpha_1}, \ldots, X_{\alpha_r}, Y_{\beta_1}, \ldots, Y_{\beta_r})~.
\end{equation} 	
Notably, this includes the popular two-sample Wilcoxon statistic (or equivalently, the Mann-Whitney statistic), given by
\begin{equation}		
    W_{m,n} = \frac{1}{mn} \sum_{i=1}^m \sum_{j=1}^n {I(X_i \leq Y_j)}~.
\end{equation} 	
one could tabulate its null distribution for given sample sizes $m$ and $n$.  
Since $W_{m,n}$ is a rank statistic, if there are no ties among the observations, then the permutation distribution can be tabled as it no longer depends on the data at all; that is, 

Similar to a test based on differences in sample means, the Wilcoxon test is exact for testing $P = Q$.
Note, however $W_{m,n}$ is unbiased for 
$$ \theta (P,Q) \equiv E_{P,Q} I \{ X_i \le Y_j \} =  P \{ X_i \le Y_j \}~,$$
and any inference about $\theta (P,Q)$ is invalid based on the permutation test using the test statistic $W_{m,n}$.
Following rejection of such a permutation, in order to claim $\theta (P,Q) > 1/2 $ (or $< 1/2$), the intuition provided earlier applies.  Therefore, if one studentizes appropriately, then asymptotic validity follows.  In particular, define the statistic
\begin{equation}
\tilde{W}_{m,n} = \frac{W_{m,n} - \frac{1}{2}}{\sqrt{\frac{N}{m}\hat{\xi}_{x} + \frac{N}{n}\hat{\xi}_{y}}},
\end{equation}
where
\begin{align}
\hat{\xi}_{x} 
& = \frac{1}{m-1} \sum_{i=1}^m \left(\frac{1}{n}(S^x_i - i) - \frac{1}{mn}\sum_{i=1}^m 
(S_i - i) \right)^2~,\quad
\hat{\xi}_{y} 
= \frac{1}{n-1} \sum_{j=1}^n \left(\frac{1}{m}(S^y_j - j) - \frac{1}{mn}\sum_{j=1}^m 
(S^y_j - j) \right)^2 ~,\nonumber
\end{align}
and $S^x_i$ and $S^y_j$ are the ranks of $X_i$ and $Y_j$ in the combined sample, respectively. That is, $\tilde W_{m,n}$ is a rank statistic and so it retains all the benefits of the usual two-sample Wilcoxon test as a rank test. The permutation test based on $\tilde W_{m,n}$ is exact under $P = Q$ and is asymptotically valid for tests of the parameter $\theta (P,Q) = 1/2.$ The test has the same high (Pitman) asymptotic relative efficiencies compared to the $t$-test, even under normality, but it is robust against non-normality. See \cite{chung2016asymptotically}, \cite{berrett2020conditional}, and \cite{kim2022minimax} for further consideration of two-sample permutation tests based on $U$-statistics.

\subsubsection{Multiple Samples}
		
Two-sample testing generalizes to the analysis of variance. Consider observations from $k$ samples, with underlying distributions $P_1 , \ldots , P_k$.  For testing 
\[H_0: \theta (P_1 )  = \cdots = \theta (P_k )\]
against the alternative
\[H_1: \theta ( P_i ) \ne \theta ( P_j )~~~{\rm for~some~}i,~j~. \]
For testing many means, the problem can be viewed as a nonparametric generalized Behrens-Fisher problem \citep{janssen1997studentized}. As before, without an appropriate choice of test statistic, permutation tests can be invalid for making inferences (such as directional claims) regarding parameters. But the same intuition suggests that, at least asymptotically, one should base a test on an asymptotically distribution-free or pivotal test statistic. To this end, let
\[T_{n} = \sum_{i=1}^k \frac{n_i}{\hat \sigma_{n,i}^2} \left( 
\hat \theta_{n,i} - \frac{ \sum_{i=1}^k n_i \hat \theta_{n,i} / \hat \sigma_{n,i}^2 }{\sum_{i=1}^k n_i/ \hat \sigma_{n,i}^2 } \right )^2~,\]
where $\hat \sigma_{n,i} \equiv \hat \sigma_{n,i}(X_{i,1}, \ldots, X_{i, n_i})$ is a consistent estimator of $\sigma_i = \sigma_i(f_{P_i})$. 
Under finite second moments, under $H_0$, $T_{n}$ converges in distribution to the Chi-squared distribution with $k-1$ degrees of freedom, and hence is asymptotically pivotal.
It follows that the permutation test based on $T_n$ asymptotically controls the Type 1 rejection probability when $H_0$ is true, while still retaining its exactness when $P_1 = \cdots = P_k$.
For details, see \cite{chung2013exact}.
 	
\subsubsection{Multivariate Permutation Test}\label{sec:mpt}
The same ideas extend to vector-valued observations as well. Suppose $X_1, \ldots, X_m$ are $d$-dimensional i.i.d. $P$ with mean vector $(\mu_1(P), \ldots, \mu_d(P) )$ and covariance matrix $\Sigma_P$, and independently, $Y_1, \ldots, Y_n$ are $d$-dimensional i.i.d. $Q$  with mean vector$(\mu_1(Q), \ldots, \mu_d(Q))$ and covariance  $\Sigma_Q $.  Consider testing
\begin{equation}
    H_0: \mu_k(P) = \mu_k(Q)~~\mbox{for all} ~~k \in [d]
    \quad\text{against}\quad
    H_1: \mu_k(P) \neq \mu_k(Q)~~\mbox{for some} ~~k \in [d]~,
\end{equation}
based on the test statistic 
\begin{equation}	
	T_{m,n} = \left(T_{m,n,1}, \ldots, T_{m,n,d} \right) = N^{1/2}\left(\frac{1}{m}\sum_{i=1}^mX_i - \frac{1}{n}\sum_{j=1}^n Y_j\right)~
\end{equation}
Since $T_{m,n}$ is not asymptotically pivotal, permutation tests may be invalid. However, we may consider a modified Hotelling's $T^2$ statistic, defined by
\begin{equation}\label{equation:smn}	
S_{m,n} = ||\hat \Sigma^{-1/2} T_{m,n}||^2 = T_{m,n}' \hat \Sigma^{-1} T_{m,n} ~,
\end{equation}
where
$\hat \Sigma = \frac{1}{p} \hat \Sigma_P + \frac{1}{1-p} \hat \Sigma_Q~$, for some estimators satisfying $\hat \Sigma_P \pto \Sigma_P$, $\hat \Sigma_Q \pto \Sigma_Q$, and $||\cdot||$ denotes the usual Euclidean norm. In this case, one can show that
\begin{equation}		 
\sup_t  | \hat R^S_{m,n} (t) -  \chi_d^2 (t) | \pto 0~,
\end{equation}
where $R^S_{m,n} (t)$ is the permutation distribution based on $S_{m,n}$. That is, again, as $S_{m,n}$ is asymptotically pivotal, permutation tests are asymptotically valid for testing $H_0$ and remain exact when $P =Q$. Alternatively, a permutation test can be based on the maximum difference between the two vectors of means. In that case, an asymptotically pivotal test statistic can be achieved using a ``bootstrap after permuting" algorithm; see \cite{chung2016permutation}.

\subsubsection{Multiple Testing}

Rather than testing a single null hypothesis, consider the multiple testing problem of simultaneously testing $s$ null hypotheses $H_1 , \ldots , H_d$.  Suppose $P$
is the unknown distribution generating the data $X$, with $P$ assumed to belong to a model $\Omega$. 
In general, a null hypothesis $H_i$ may be described as $P \in \omega_i$, where $\omega_i \subset \Omega$ is some subset of the full model.
When testing multiple hypotheses, the goal is to detect which null hypotheses are false.  In order to control for false rejections,  the classical approach 
is to control the familywise error rate (FWER), defined as the probability of at least one false rejection, at some level $\alpha$.  
Alternative measures of error control, such as the false discovery rate,
are reviewed in Chapter 9 of \cite{lehmann:romano:tsh:2022}. 

Since $p$-values of individual hypotheses may be constructed using randomization tests, as in (\ref{eq: p value conserv}),  any method that combines $p$-values may be considered to test multiple hypotheses.  Examples are the methods of \cite{holm:1979} and 
Benjamini Yekutieli method \cite{benjamini:yekutieli:2001}.  However, methods
based on individual (or marginal) $p$-values are generally conservative.

One can sometimes derive tests that are not conservative.  Here, we focus on FWER control. 
One approach is to use  the previous tests of multivariate parameters in conjunction with the {\it closure method} to derive tests of multiple hypotheses \citep{romano2011consonance}. Such tests control the familywise error rate (FWER), and implicitly account for the dependence among the test statistics (and hence offer greater power than Bonferroni/Holm type methods). Moreover, they control FWER exactly in finite samples (when an appropriate  randomization hypothesis holds) and have asymptotic validity otherwise.

To elaborate further, we first describe the closure method.  For any $K \subset \{ 1, \ldots , d \}$, let $H_K$ denote the joint (but single)  null hypothesis that all $H_i$ with $i \in K$ are true; that is, $P \in \bigcap_{i \in K} {\omega_i}$. Suppose $\phi_K$ is a test of $H_K$ that controls the usual probability of a Type 1 error at level $\alpha$.  Now, define the multiple testing method that rejects $H_i$ if and only if $H_K$ is rejected whenever $K$ contains $i$.    Then, the FWER is controlled at level $\alpha$.  

The proof is quick.  For any $P \in \Omega$, suppose $I = I(P)$ is the set of indices $i$ for which $P \in \omega_i$; that is, $I$ is the set of indices of true null hypotheses.   Then, the probability of a false rejection
is bounded above by the chance that $H_I$ is rejected because in order to reject $H_i$ for $i \in I$, $H_I$ must be rejected  (as well as many other $H_K$s). But, by assumption the joint test $\phi_I $ rejects $H_I$ with probability $\le \alpha$.

To make these ideas concrete, consider the setting of Section \ref{sec:mpt} of testing equality of many means.
Instead of a joint test of many means, whose rejection would claim at least one difference is significant without specifying which ones,  we are now interested in testing 
individual differences; that is, let
$$H_i : \mu_i (P) = \mu_i (Q)~,
$$
so that $H_i$ specifies the $i$th components of the mean vectors of the two populations are equal.
Then, for any $K \subset \{ 1 , \ldots , d 
\}$, $H_K$ specifies equality of means for $i \in K$.
Construct the joint test of $H_K$ using the test statistic $S_{m,n} $ given in
(\ref{equation:smn}), except use only components $i$ of the data with $i \in K$.
This test has finite-sample validity under the randomization hypothesis that 
the joint distributions of $X_j $ and $Y_j$ are equal, but more generally when the joint distributions of the components of $X_j$ and $Y_j$ with $i \in K$ are equal.
Furthermore, the randomization test is asymptotically valid otherwise, as described
in Section \ref{sec:mpt}.  Now use these joint tests when applying the closure method.  That is $H_i$ is rejected if and only if all permutation tests of $H_K$ such that $K$ contains $i$ are rejected.  The result is a multiple testing procedure that controls the FWER exactly when the joint distributions of  the components $i \in I$ of  $X_j$ and $Y_j$ are equal, but asymptotically in general (meaning when the means are equal but the distributions may differ). 
For an application of this methodology, \cite{chung2016permutation} revisit the data in
\cite{charness}, who examine the effects of exercise on multiple biometric measures.

The use of permutation and other resampling-based tests for multiple testing originated in \cite{westfall:young:1993}, under a condition called subset pivotality.
This condition was shown to be unnecessary in \cite{romano:wolf:2005jasa}.
Some recent work that uses permutation methodology in multiple testing is
\cite{hemerik2019permutation} and \cite{vesely2023permutation}.

\section{Experiments}\label{section:experiments}

The analysis of experimental data (or data from randomized controlled trials) is particularly well-suited to randomization inference.  To see why, consider an experiment in which units are first sampled i.i.d.\ from some distribution $P$.  Such a sampling scheme is sometimes referred to in the literature on experiments as sampling from a ``superpopulation'' in an effort to distinguish it from alternative ``finite population'' sampling schemes; see below for some discussion on the relationship between these different sampling schemes.  For each unit $i$, denote by $Y_i(1)$ the potential outcome under treatment, by $Y_i(0)$ the potential outcome under control, and by $Z_i$ observed, baseline covariates.  Further denote by $D_i$ the treatment status of the $i$th unit.  For each unit $i$, the observed data is denoted by $X_i = (Y_i, D_i,Z_i)$, where the observed outcome $Y_i$ satisfies 
\begin{equation} \label{eq:obstopot}
Y_i = Y_i(1) D_i + Y_i(0)(1 - D_i)~.
\end{equation}
In what follows, for a generic random variable indexed by $i$, $W_i$, it will be convenient to let $W^{(n)} = (W_1, \ldots, W_n)$.  Using this notation, we will make the following assumption on the treatment assignment rule 
\begin{equation} \label{eq:unconfoundedness}
(Y^{(n)}(1), Y^{(n)}(0)) \independent D^{(n)} | Z^{(n)}~.
\end{equation}
In words, such an assumption requires that the (joint) distribution of treatment status only depends on the observed, baseline covariates.  In practice, this distribution is known to the researcher in the context of an experiment.

Before proceeding, we describe some common treatment assignment rules.  Perhaps the simplest example of such a rule is what is referred to sometimes as simple random sampling, in which $D^{(n)} \independent Z^{(n)}$ and $D_i, i = 1, \ldots, n$ are i.i.d.\ $\sim$ Bernoulli$(q)$ for some known $0 < q < 1$.  Another common treatment assignment rule is what is known as complete randomization in which again $D^{(n)} \independent Z^{(n)}$ and $D^{(n)}$ is uniformly distributed over vectors $d^{(n)} = (d_1, \ldots, d_n)$ in which each $d_i \in \{0,1\}$ and $\sum_{1 \leq i \leq n} d_i = m$ for some known $0 < m < n$.  A somewhat more complicated treatment assignment scheme is stratified block randomization.  As the name suggests, in this treatment assignment scheme units are first divided into $s$ strata according to $S_i = S(Z_i)$, where $S : \text{supp}(Z_i) \rightarrow \{1, \ldots, s\}$, and then, within each stratum (and independently across strata), units are assigned to treatment according to complete randomization with $m \approx q \times n(s)$, where $n(s) = \sum_{1 \leq i \leq n} I\{S_i = s\}$.  Finally, a fourth treatment assignment scheme that is often employed is what is known as matched pair designs.  In such designs, units are first paired according to observed, baseline covariates, and then, within each pair, one unit is assigned to treatment and the other to control with equal probability.

\subsection{Testing ``Strong'' Null Hypotheses} \label{sec:strong}

For any of the randomization schemes described above, it is possible to devise a randomization test that is exact for the following null hypothesis: 
\begin{equation} \label{eq:sharp}
H_0 : Y_i(1) | Z_i \stackrel{d}{=} Y_i(0) | Z_i~.
\end{equation}
This null hypothesis may be viewed as positing a ``strong'' sense in which the treatment has no effect on t\textbf{}he outcome of interest.  

To describe this test, note that for any of the four randomization schemes described above, there exists a group $\mathbf G_{Z^{(n)}}$ that preserves the distribution of treatment status in the sense that 
\begin{equation} \label{eq:treatmentinvar}
g D^{(n)} | Z^{(n)} \stackrel{d}{=} D^{(n)} | Z^{(n)} \text { for any } g \in \mathbf G_{Z^{(n)}} ~.
\end{equation}
For concreteness, we briefly describe these transformations for each of the treatment assignment schemes above.  In the context of simple random sampling and complete randomization, one possible choice of $\mathbf G_{Z^{(n)}}$ is simply $\mathbf G_n$, the set of all permutations of $n$ elements, and $g D^{(n)}$ is defined for any $g \in \mathbf G_n$ to be $(D_{g(1)}, \ldots, D_{g(n)})$.  In the context of stratified block randomization, a natural choice of $\mathbf G_{Z^{(n)}}$ is the subgroup of $G_n$ that only permutes units within a common stratum, i.e., $$\mathbf G_{S^{(n)}} = \{g \in \mathbf G_n : S_{g(i)} = S_i \text{ for all } 1 \leq i \leq n\}~,$$ and $gD^{(n)}$ is defined in the same way as before.  Finally, in the context of matched pair designs, essentially the same idea applies with the strata being understood as the pairs.  In order to describe the group more formally, we require some further notation.  For convenience, suppose $n$ is even, i.e., $n = 2k$ for some integer $k$.  For such $n$, we may denote the $k$ pairs by $(\pi(2j-1),\pi(2j)), j = 1, \ldots, k$, where $\pi = \pi_{Z^{(n))}}$ is a permutation of $n$ elements (that, importantly, may depend on $Z^{(n)}$).  Using this notation, the natural choice of $\mathbf G_{Z^{(n)}}$ is the subgroup of $G_n$ that only permutes units within a pair, i.e., $$\mathbf G_{Z^{(n)}} = \{g \in \mathbf G_n : \{g(\pi(2j-1)), g(\pi(2j))\} = \{\pi(2j-1), \pi(2j)\} \text{ for all } 1 \leq j \leq p\}~,$$ and $gD^{(n)}$ is again defined in the same way as before. 

We now argue that under \eqref{eq:sharp}, the distribution of the observed data $X^{(n)}$ is invariant with respect to the transformations in $\mathbf G_{Z^{(n)}}$ satisfying \eqref{eq:treatmentinvar} in the sense that 
\begin{equation} \label{eq:keyinvar}
g X^{(n)} \stackrel{d}{=} X^{(n)}~,
\end{equation}
where $g X^{(n)} = (Y^{(n)}, gD^{(n)}, Z^{(n)})$.  In order to establish this equality in distribution, we argue that $Y^{(n)} \independent D^{(n)} | Z^{(n)}$.  To see this, let $A^{(n)} = \prod_{1 \leq i \leq n } A_i$ for arbitrary intervals $A_i$ and $Y^{(n)}(d^{(n)}) = (Y_1(d_1), \ldots, Y_n(d_n))$.  With this notation in mind, note that 
\begin{eqnarray*}
P\{Y^{(n)} \in A^{(n)} | D^{(n)} = d^{(n)}, Z^{(n)}\} &=& P\{Y^{(n)}(d^{(n)}) \in A^{(n)} | D^{(n)} = d^{(n)}, Z^{(n)}\} \\
&=& P\{Y_1(d_1) \in A_1, \ldots,  \in Y_n(d_n) \in A_n | D^{(n)} = d^{(n)}, Z^{(n)}\} \\
&=& P\{Y_1(d_1) \in A_1, \ldots,  \in Y_n(d_n) \in A_n | Z^{(n)} \} \\
&=& \prod_{1 \leq i \leq n } P\{Y_i(d_i) \in A_i | Z_i\}~,
\end{eqnarray*}
where the first equality exploits \eqref{eq:obstopot}, the second exploits the definitions of $A^{(n)}$ and $Y^{(n)}(d^{(n)})$ given above, the third equality exploits \eqref{eq:unconfoundedness}, and the final equality exploits i.i.d.\ sampling.  Under \eqref{eq:sharp}, we see that this final quantity does not depend on $d^{(n)}$, from which the desired conditional independence property follows.  We now have immediately that \eqref{eq:treatmentinvar} implies $g X^{(n)} | Z^{(n)} \stackrel{d}{=} X^{(n)} | Z^{(n)}$ under \eqref{eq:sharp}, from which the desired unconditional equality in distribution in \eqref{eq:keyinvar} follows as well.

The test is now constructed in the usual way.  For any test statistic $T_n = T_n(X^{(n)})$ such that large values provide evidence against $T_n$, we can construct a suitable critical value with which to compare it as $\hat r_n^{-1}(1 - \alpha)$, defined in (\ref{equation:juri}), where $$\hat R_n(t) = \frac{1}{|\mathbf G_{Z^{(n)}}|} \sum_{g \in \mathbf G_{Z^{(n)}}} I\{ T_n(gX^{(n)}) \leq t \}~.$$ Having established \eqref{eq:keyinvar}, it is now straightforward to modify the proof of Theorem \ref{theorem:1} to show that the test that rejects $H_0$ whenever $T_n$ exceeds $T^{(k)}_n(X^{(n)})$ is level $\alpha$ in finite samples.

\begin{remark}
The development above presumes that there are no spillovers in the sense that potential outcomes for unit $i$ are only indexed by the treatment status of unit $i$.  When this is not the case, potential outcomes are now indexed by a vector-valued argument that specifies the treatment of each unit.  Randomization inference may be used essentially verbatim to test a suitably modified version of the ``strong'' null hypothesis that asserts that the (conditional) distribution of each potential outcome is invariant with respect to this argument.  More interestingly, however, by restricting attention to suitable subsets of the data, it is possible to test somewhat less restrictive null hypotheses.  For a development of such an idea in a finite population context, see \cite{basse2019randomization} and \cite{puelz2022graph}. See also \cite{athey2018exact}, \cite{li2019randomization}, \cite{xu2021randomization}, and \cite{basse2024randomization} for further consideration randomization inference for problems incorporating spillovers and interference.
\end{remark}

\begin{remark}
Even if it is not possible to exhibit a group satisfying \eqref{eq:treatmentinvar}, it is possible to construct tests of the null hypothesis in \eqref{eq:sharp} by exploiting knowledge of the distribution of $D^{(n)} | Z^{(n)}$ and the logic of Remark \ref{rem:samplingsubset}.  To see this, for $b = 2, \ldots, B$, let $D^{(n),b}, b = 2, \ldots, B$ be i.i.d.\ $\sim D^{(n)} | Z^{(n)}$ and define $X^{(n),b} = (Y^{(n)}, D^{(n),b}, Z^{(n)})$ let $X^{(n),1} = X^{(n)}$.  It is straightforward to argue that $\{T_n(X^{(n),b}) : 1 \leq b \leq B\}$ is exchangeable under null hypothesis in \eqref{eq:sharp}.  The rest of the argument in Remark \ref{rem:samplingsubset} now applies verbatim and we can employ the same construction to yield a test that is level $\alpha$ in finite samples for testing the null hypothesis in \eqref{eq:sharp}. Conditional randomization tests were introduced by \cite{rosenbaum1984conditional}, and were further developed by \cite{candes2018panning} and \cite{berrett2020conditional}
\end{remark}

\subsection{Testing ``Weak'' Null Hypotheses} \label{sec:weak}

We are, of course, often not interested in testing the null hypothesis in \eqref{eq:sharp}, but rather the null hypothesis 
\begin{equation} \label{eq:weaknull}
H_0: E[Y_i(1) - Y_i(0)] = 0~.
\end{equation}
Such null hypotheses are sometimes referred to as ``weak'' null hypotheses to distinguish them from the ``strong'' null hypotheses considered in the preceding section; see e.g., \cite{chung2017randomization} for further discussion. For such a null hypothesis, it is natural to consider tests that reject for large values of $|t_n|$, where $T_n(X^{(n)}) = \sqrt n (\hat \mu_n(1) - \hat \mu_n(0))$, $\hat \mu_n(d) = \frac{1}{n_d} \sum_{i=1}^n \mathbb{I}\{D_i = d\}Y_i$, and $n_d=\sum_{i=1}^n \mathbb{I}\{D_i = d\}$.  In order to describe the large-sample behavior of $T_n(X^{(n)})$, it is useful to specialize to a specific treatment assignment rule.  Below we focus on matched pair designs, following the treatment in \cite{bai2019inference}.  For some results related to other treatment assignment schemes, especially stratified block randomization, see \cite{bugni2018inference, bugni2019inference}. For further consideration of tests of weak null hypotheses in randomized experiments, see \cite{wu2021randomization}, \cite{zhao2021covariate}, and \cite{heckman2023dealing}.

Under weak assumptions, \cite{bai2019inference} establish that $T_n \stackrel{d}{\rightarrow} N(0,V)$, where $$V =  E[\text{Var}[Y_i(1)|Z_i]] + E[\text{Var}[Y_i(0)|Z_i]] +  \frac{1}{2} E\left [\left (E[Y_i(1)|Z_i] - E[Y_i(0)|Z_i] \right )^2\right ]~.$$  The main requirements underlying this result are that pairs are constructed so that units within a pair are suitably close in terms of their observed, baseline covariates, specifically 
\begin{equation} \label{eq:pairsclose}
\frac{1}{n} \sum_{1 \leq j \leq p} |Z_{\pi(2j)} - Z_{\pi(2j-1)}| \stackrel{P}{\rightarrow} 0~,
\end{equation}
and that $E[Y_i(d) | Z_i]$ is sufficiently well behaved.  \cite{bai2019inference} provide algorithms that ensure \eqref{eq:pairsclose} satisfied; in particular, it suffices to choose $\pi$ to minimize the lefthand-side of \eqref{eq:pairsclose} under mild moment restrictions.  \cite{bai2019inference} assume that $E[Y_i(d)|Z_i]$ is Lipschitz, but it is enough to assume it is simply integrable by suitably approximating integrable functions with Lipschitz functions; see \cite{cytrynbaum2023optimal} for such an argument assuming $E[Y_i(d)|Z_i]$ is square integrable.

By contrast, the randomization distribution of $T_n(X^{(n)})$ is simply the distribution
$$\frac{1}{\sqrt n} \sum_{1 \leq j \leq p} \epsilon_j (Y_{\pi(2j)} - Y_{\pi(2j-1)})~,$$
conditional on $X^{(n)}$, where $\epsilon_j, j = 1, \ldots, n$ are i.i.d.\ Rademacher random variables, i.e., taking on values $\pm 1$ with equal probability each.  \cite{bai2019inference} show that $$\hat R_n(t) \stackrel{P}{\rightarrow} \Phi(t/\tau)~,$$ where $$\tau^2 = E[\text{Var}[Y_i(1)|Z_i]] + E[\text{Var}[Y_i(0)|Z_i]] + E\left [\left (E[Y_i(1)|Z_i] - E[Y_i(0)|Z_i] \right )^2\right ]~.$$  Since $\tau^2 \geq V$, it follows that the randomization test of the null hypothesis in \eqref{eq:weaknull} based on $t_n(X^{(n)})$ is generally conservative.  \cite{bai2019inference} show that qualitatively similar results hold for the usual two-sample $t$-test and the paired $t$-test.  

As in the preceding sections, however, it is possible to restore asymptotic exactness of the randomization test by applying it to an appropriately Studentized version of $T_n(X^{(n)})$.  As a first step towards this, \cite{bai2019inference} develop a consistent estimator of $V$.  A key challenge in doing so is the estimation of quantities like $E[E[Y_i(1)|Z_i]^2]$ because the natural estimator of this quantity would involve two independent observations of $Y_i(1)$ conditional on $Z_i$ and, by construction, only one such observation is available.  \cite{bai2019inference} show, however, that it is possible to estimate such quantities under a suitable strengthening of \eqref{eq:pairsclose} that ensures that units in adjacent pairs are also close in terms of their observed, baseline covariates.  If we denote by $\hat V_n$ the resulting estimator of $V$, then \cite{bai2019inference} show that the randomization test of the ``weak'' null hypothesis in \eqref{eq:weaknull} with $S_n(X^{(n)}) = |T_n(X^{(n)})|/\sqrt{\hat V_n}$ is asymptotically exact, while remaining level $\alpha$ in finite samples for the ``strong'' null hypothesis in \eqref{eq:sharp}.

\begin{remark}
In the preceding discussion, we have focused on testing \eqref{eq:weaknull}, but it is straightforward to modify the procedure so as to test the null hypothesis that specifies instead 
\begin{equation} \label{eq:weaknullnonzero}
E[Y_i(1) - Y_i(0)] = \theta_0
\end{equation}
for some pre-specified value $\theta_0$ that need not equal zero.  To do so, we can simply ``pre-process'' the data by replacing $Y_i$ with $Y_i - \theta_0 D_i$.  The data transformed in this way now satisfies \eqref{eq:weaknull} whenever the null hypothesis of interest \eqref{eq:weaknullnonzero} holds.  In this way, using test inversion, it is possible to construct confidence intervals for $E[Y_i(1) - Y_i(0)]$.
\end{remark}

\begin{remark}
As mentioned above, our discussion has focused on a ``superpopulation'' sampling framework.  An alternative sampling framework is sampling from a ``finite population.''  In such a framework, $n$ units are sampled without replacement from a finite population of $N$ units defined by $\{(y_i(1), y_i(0), z_i) : 1 \leq i \leq N \}$. The situation in which $n = N$ is sometimes referred to as a ``design-based'' setting because the only remaining source of uncertainty is from the design, by which this literature means the way in which treatment status was assigned, but the framework permits $n < N$ as well.   In such settings, the randomization test described in Section \ref{sec:strong} is level $\alpha$ in finite samples for testing the following counterpart to the ``strong'' null hypothesis in \eqref{eq:sharp}: $$H_0: y_i(1) = y_i(0) \text{ for all } 1 \leq i \leq N~.$$  This null hypothesis is sometimes referred to as ``sharp'' because it permits one to impute the observed data for any possible value of $D^{(n)}$.  In a finite population, the corresponding counterpart to the ``weak'' null hypothesis in \eqref{eq:weaknull} is 
\begin{equation} \label{eq:weakfinite}
H_0 : \frac{1}{N} \sum_{1 \leq i \leq N} \left ( y_i(1) - y_i(0) \right )  = 0~.
\end{equation}
Randomization tests may again be used to test the null hypothesis in \eqref{eq:weakfinite} as well, but there are some subtle differences in the analysis.  A key feature is that generally these tests will not be exact even asymptotically, though generally are at least asymptotically conservative; see \cite{ding2017paradox} and \cite{wu2021randomization}.  This phenomenon stems from the limiting behavior of quantities like $T_n(X^{(n)})$ when sampling from a finite population.   In particular, their limits in distribution may involve features of $\{y_i(1) - y_i(0) : 1 \leq i \leq N\}$ that are not consistently estimable using the data from the experiment.  This feature typically remains present unless $n/N \rightarrow 0$.  See, e.g., \cite{lehmann:romano:tsh:2022} for an analysis under the assumption that treatment is assigned using complete randomization.  For further discussion of the relationship between these two sampling schemes, we additionally refer the reader to \cite{imbens2015causal}, \cite{abadie2020sampling}, \cite{bai2024survey}, and the references therein.
\end{remark}

\begin{remark}
Throughout this section, we have restricted attention to experiments with a single treatment and a control.  Many of the ideas generalize naturally to settings in which there are multiple treatments. See \cite{bugni2019inference},  \cite{bai2024inference}, and \cite{bai2024survey}.  The latter reference surveys a number of topics in the broader literature on the analysis of experiments that are omitted from our discussion, including the broader benefits of stratification in terms of reducing {\it ex post} bias of estimators, cluster-level randomized experiments, and regression adjustment in experiments.  
\end{remark}

\section{Correlation and Regression}\label{section:corr}

We now consider tests concerning correlation and regression. Here, randomization tests will be exact for null hypotheses involving independence restrictions. Asymptotic correctness for tests of correlation, or for tests involving regression coefficients, can be obtained with appropriate choices of test statistic. 

\subsection{Correlation} Assume $\left(X_1, Y_1 \right),\ldots, \left(X_n, Y_n \right)$ are i.i.d.\ according to a joint distribution $P$ with (non-degenerate) marginal distributions $P_X$ and $P_Y$.  Define $X^{(n)} = (X_1,...,X_n)$ and $Y^{(n)} = (Y_1,...,Y_n)$. Consider the problem of testing the null hypothesis of independence, given by
\begin{equation}\label{eq: independence}
H_0: P = P_X \times P_Y.
\end{equation}
Define $\bf{G}_n$ to be the set of all permutations $\pi$ of $\left\{ 1,...,n \right\}$. As before, the permutation distribution of any given test statistic $T_n\left(X^{(n)}, Y^{(n)} \right)$ is given by
\[
\hat R_n (t) = \frac{1}{n!} \sum_{\pi \in \bf{G}_n} I \left\{ T_n(X^{(n)}, Y^{(n)}_{\pi})  \leq t \right\}
\] 
Since the randomization hypothesis holds for testing $H_0$, an exact permutation test can be constructed.

On the other hand, consider instead the null hypothesis 
\begin{equation}\label{eq: correlation}
H_0: \rho(P) = 0~,
\end{equation}
where $\rho = \rho(P)= \text{corr} (X_1, Y_1)$ denotes the correlation between $X_1$ and $Y_1$. The normalized sample correlation
\[
\sqrt{n} \hat \rho_n( X^{(n)}, Y^{(n)}) =  \frac{\sqrt{n} \sum_{i=1}^n X_{i} Y_{i} - n \bar X_n \bar Y_n}{\sqrt{\sum_{i=1}^n (X_{i} - \bar X_n)^2 \sum_{i=1}^n (Y_{i} - \bar Y_n)^2}}~.
\]
is a natural choice of test statistic. If rejection of the null hypothesis $H_0$ is accompanied by the claim that $\rho$ is positive when $\hat \rho_n$ is large and positive, then such claims can have large Type 3, or directional, error rates. To see this, note that under $H_0$, the observations $X_i$ and $Y_i$ may be dependent, and the distribution of the test statistic may not be the same under all permutations of the data. That is, the randomization hypothesis fails and the test is not guaranteed to be level $\alpha$, even asymptotically.  

Under the null hypothesis $H_0$,  if $E (X_1)^2 < \infty$, $E (Y_1)^2 < \infty$ and $E (X_1 Y_1)^2 < \infty$, then the sampling distribution of $\sqrt{n} \hat \rho_n(X^n, Y^n)$ is $N(0, \tau^2 (P))$, where
\begin{equation} \label{eq:1}
\tau^2 = \tau^2(P) = \frac{\mu_{2,2}}{\mu_{2,0}\mu_{0,2}}~,\quad\text{for}\quad\mu_{r,s} = E \left[ (X_1 - \mu_X)^r (Y_1 - \mu_Y)^s \right]~,
\end{equation}
and $\mu_X$ and $\mu_Y$ are the means of the $X_i$ and $Y_i$, respectively. 
However, the permutation distribution is asymptotically $N(0,1)$ (with probability one).  As in the two-sample problems considered in Section \ref{section:two}, there is a problem of mismatched asymptotic variances. 

Before considering a fix, it is instructive to consider why the approximation permutation distribution is standard normal.  The intuition is as follows.
Applying a permutation to one coordinate of the data effectively destroys the dependence between pairs.
Therefore,  the behavior of the permutation distribution under a general joint distribution $P$ should not be too different from when observations
are sampled from the joint distribution $\tilde P$ where $X$ and $Y$ are independent, but with the same marginal distributions as $P$; that is,
$\tilde P = P_X \times P_Y$.
In such case, the randomization hypothesis holds, and using the intuition from Remark \ref{remark:intuition}, the permutation distribution should be approximately the limiting  distribution of the unconditional distribution under $\tilde P$, in which case the limiting approximation is, by (\ref{eq:1}), $V ( \tilde P ) = 1$. On the other hand, if $X_i$ and $Y_i$ are uncorrelated with joint distribution $P$ (that is not necessarily independent), then
\[
0 \leq  V (P) \equiv \frac{E \left[ (X_i - \mu_{X})^2 (Y_i - \mu_{Y})^2 \right]}{ \sigma^{2}_{X} \sigma^{2}_{Y} } \leq \infty
\]
and these bounds can be attained in the sense that there exists a joint distribution of $(X_1, Y_1)$ where $\text{cov}(X_1, Y_1) = 0$, but this ratio is 0, and likewise for which it is $\infty$.
Thus, the probability of a Type 1 error can be near 1. Even more troubling, this discrepancy can lead to large Type 3 (directional) error rate if one is interested in deciding the sign of the correlation based on the sample correlation. 

To remedy these problems, the same intuition applies, and one should use a test statistic that is asymptotically pivotal.  In this case, the test statistic can be studentized by 
\[
\hat V_n = \sqrt{\frac{\hat \mu_{2,2}}{\hat \mu_{2,0} \hat \mu_{0,2}}} \quad\text{where}\quad\hat \mu_{r,s} = \frac{1}{n} \sum^{n}_{i=1} (X_i - \bar X_n)^r(Y_i - \bar Y_n)^s
\]  
are the sample central moments.  The studentized correlation statistic defined by $S_n = \sqrt{n} \hat \rho_n / \hat V_n$ is asymptotically pivotal in the sense that it is asymptotically distribution free whenever the underlying distribution satisfies $H_0$. Because $S_n$ is a pivotal statistic, the true sampling distribution of $S_n$ under $P$ has the same asymptotic behavior as the true sampling distribution of $S_n$ under $P_X \times P_Y$.  When the randomization hypothesis is satisfied, the permutation test using the statistic $S_n$ is exact under $P_X \times P_Y$, and therefore, the permutation distribution should asymptotically approximate the true sampling distribution under $P_X \times P_Y$. Formally, the following result from \cite{diciccio2017robust} holds.

\begin{theorem} \label{theorem:2}
Assume $\left(X_1, Y_1 \right), ..., \left(X_n, Y_n \right)$ are i.i.d. according to $P$ such that $X_1$ and $Y_1$ are uncorrelated but not necessarily independent.  Also assume that $E(X_1^4) < \infty$ and $E(Y_1^4) < \infty$.  The permutation distribution $\hat R^{S_n}_n (t)$ of $S_n = \sqrt{n} \hat \rho_n / \hat \tau_n$ satisfies
\[
\lim_{n \rightarrow \infty} \sup_{t \in \mathbb{R}} \left| \hat R^{S_n}_n (t) - \Phi(t) \right| = 0
\]
almost surely.
\end{theorem}
\noindent Consequently, if $\sqrt{n} \hat \rho_n$ is studentized by $\hat V_n$, then the quantiles of the permutation distribution and the true sampling distribution converge almsot surely to the corresponding quantiles of the standard normal distribution.  The permutation test using the studentized statistic is appealing because it retains the exactness property under $P_X \times P_Y$ but is also asymptotically level $\alpha$ under $P$.  

\subsection{Regression}
Many of the same ideas generalize to consideration of a coefficient in a linear regression. Consider i.i.d.\ $(X_{i}, Y_{i})$ pairs following a simple univariate linear regression model
\[
Y_i = \alpha + \beta X_i + \epsilon_i,\quad i=1,...,n~,
\]
where $X_i \in \mathbb{R}$ and $\epsilon_i$ are errors with mean zero and variance $\sigma^2$.  (For the moment, assumptions on the joint distribution of $X_i$ and $\epsilon_i$ are not specified, but will be described below).  To test the hypothesis 
\[
H_0: \beta = 0,
\]
it is natural to base a test on the sample correlation  $\hat \rho_n$  
or the least squares estimator $\hat \beta_n$.  If the $X_i$'s are independent of the $\epsilon_i$'s (and therefore independent of the $Y_i$'s under $H_0:\beta = 0$), then an exact permutation test can be performed by permuting the $X_i$'s.  A permutation test may not be exact if the predictors and errors are uncorrelated (but dependent), however, following the previous results, studentizing the correlation coefficient leads to an asymptotically level $\alpha$ test.  

\cite{hartigan1970exact} applied sign changes to residuals under the assumption of symmetrically distributed errors; also see \cite{theil1950rank}, \cite{brown1951median} and \cite{daniels1954distribution}.
Permutation tests have been considered in more complex regression problems. Of course, multiple regression can be considered, as well as sub-vector inference \citep{d2024robust} and accommodation of heteroskedasticity. \cite{freedman:lane:1980} permute residuals after model fitting. This approach need not provide error control, but \cite{diciccio2017robust} provide a modification that does. One may also base a randomization test based on sign changes of residuals (under symmetry of errors), which is closely related to wild bootstrap. An empirical investigation of such techniques, in the high-dimensional case, is given in \cite{hemerik2020robust}. \cite{lei2021assumption} propose a method for exact inference under exchangeable errors and give excellent review of related literature. 
For a nice review of subvector inference in regression based on randomization,
see \cite{anderson2001permutation}, who review \cite{freedman:lane:1980},
\cite{ter1992permutation} and \cite{kennedy1996randomization}.
See also \cite{guo2023invariance}, \cite{young2023consistency}, \cite{young2024asymptotically}, and \cite{pouliot2024exact} for several recent contributions to this subject.

\section{Permutation Tests in Time Series}
 
Assume  $X_1, \, \dots, \, X_n$ are observed from a time series. Of course, many parametric testing methods exist for testing the structure of the underlying stochastic process, but the assumptions for these tests are restrictive and the null hypothesis is often incorrectly specified.  However, nonparametric inferential tools for certain structural features  can be based on randomization tests. Such features include the presence (or not) of autocorrelation and trend.

\subsection{Testing Time Series Structure}
    
Assume data  $X_1, \, \dots, \, X_n$ comes  from a strictly stationary (weakly dependent)  time series.  A basic question is:
Are the observations independent?  (In a time series setting, this may seem clearly not true; however, tests of fit of many time series models are often based on residuals that are hopefully independent, at least approximately.) A frequently-used analogue for testing independence is the related question: Do the observations have zero autocorrelation? Let $H_r$ be the null hypothesis $H_r: \rho(1) = \dots = \rho (r) = 0$, where $\rho(k)$ is the $k$-th order autocorrelation. Such a problem may arise, for example, in testing the efficient market hypothesis in finance; see e.g., \cite{fama1970efficient}.

Examples of tests of $H_r$ include the Ljung-Box test  and the Box-Pierce test. These tests often make parametric assumptions on the sequence $X_1, \, \dots, \, X_n$. As such, the tests are, in general, not  exact for finite samples and may not be even asymptotically valid for controlling Type 1 error.
 
Under the more stringent null hypothesis $\bar{H}$ that $X_1, \, \dots, X_n$ are i.i.d., the randomization hypothesis holds with respect to the permutation group. Therefore, one can construct exact finite-sample permutation tests. We have encountered this setting in Section \ref{sec:hothand}. Here, the null hypotheses $\bar{H}$ and $H_r$ are quite different.  Permutation tests based on sample autocorrelations accompanied by claims about the population autocorrelation parameters can lead to issues with Type 1 error and Type 3 error control. This is the same phenomenon discussed in Sections \ref{section:two} and \ref{section:corr}. 

Despite this, permutations need not be abandoned. Our goal is to construct valid permutation tests of the null hypothesis $H^{(k)}: \rho(k) = 0$ that are
exact under the i.i.d.\ assumption, and  asymptotically valid for a large class of weakly dependent sequences. To this end, let 
\begin{align}
    \hat{\rho}_n(k) = \hat{\rho}_n \left(X_1, \, \dots, \, X_n ; k \right) =\frac{\frac{1}{n} \sum_{i=1} ^{n-k} (X_i-\bar{X}_n)( X_{i+k}- \bar{X}_n)}{\hat{\sigma}_n ^2  } \, \, ,
\end{align}
where $\bar{X}_n$ is the sample mean, and $$\hat{\sigma}_n^2 = \frac{1}{n} \sum_{i=1} ^n (X_i - \bar{X}_n)^2$$ is the sample variance.  For simplicity, let $\hat \rho_n = \hat \rho_n (1)$ and $\rho = \rho (1)$.

Before describing the behavior of the permutation distribution  based on the test statistics $\sqrt{n} \hat \rho_n$, we first consider its usual limiting distribution.
Assume that $\{X_n \}$ is a stationary, $\alpha$-mixing sequence, with mixing coefficients $\alpha_X ( \cdot )$ that satisfy, for some $\delta > 0$, $\sum_{n \geq 1} \alpha_X (n) ^{\frac{\delta}{2 + \delta}} < \infty$.
Also, assume $E \left[ \left|X_1\right| ^{8 + 4 \delta} \right] < \infty$. By \citeauthor{ibragimov1962theory}'s \citeyearpar{ibragimov1962theory} Central Limit Theorem,
$$\sqrt{n}(\hat{\rho}_n - \rho) \overset{d} \to N\left(0, \, \gamma_1^2 \right)~~,$$
for some variance $\gamma_1 ^2$. To describe $\gamma_1^2$, let
\begin{align}
\kappa^2 &= \text{Var} \left(X_1 ^2 \right) + 2\sum_{k \geq 2} \text{Cov} \left( X_1 ^2, \, X_k ^2 \right)~,\nonumber\\
\tau_1 ^2 &= \text{Var} \left(X_1X_2\right) + 2\sum_{k \geq 2} \text{Cov} (X_1X_2, \, X_k X_{k+1} )~,\nonumber\\
\nu_1 &= \text{Cov} \left( X_1 X_2, \, X_1 ^2 \right) + \sum_{k \geq 2} \text{Cov} \left( X_1^2, \, X_{k} X_{k+1} \right) + \sum_{k \geq 2} Cov \left(X_1 X_2,\, X_k ^2 \right)~,\quad\text{and}\nonumber\\
\gamma_1^2 &= \frac{1}{\sigma^4} \left( \tau_1^2  - 2 \rho_1 \nu_1 + \rho_1 ^2 \kappa^2 \right)~.\nonumber
\end{align}
Note that, while the expressions are messy,  $\gamma_1 = 1$ if the data are i.i.d.

How should we understand the permutation distribution in this setting?  Since permuting randomly changes the order of the observations, 
the behavior of the permutation distribution should be similar to the situation when all observations are i.i.d. according to a distribution $P_1$, where $P_1$ is the marginal distribution of each observation in the underlying stationary time series.  But, in the i.i.d.\ situation,
by Remark \ref{remark:intuition},
the permutation distribution should behave like the true unconditional distribution, which is $N( 0, \gamma_1^2 )$ with $\gamma_1 = 1$.
Indeed, the permutation distribution of $\sqrt{n} \hat{\rho}_n$  satisfies
\begin{equation*}
    \sup_{t \in {\bf R}} | \hat{R}_n (t) - \Phi(t) |  \overset{p} \to 0 \, \, .
\end{equation*}
Therefore, permutation tests may be asymptotically invalid, due
to the problem of mismatched variances: the true limiting variance is $\gamma_1^2$ while the permutation distribution has a variance approximately equal to one.
The solution again is to use an asymptotically pivotal statistic,
which can be obtained by studentizing.
An appropriate studentized  test statistic is presented 
in \cite{romano2022permutation}, accompanied by the limiting behavior of the permutation test.
Permutation tests have been extended to regression problems with correlated errors in \cite{romano2024least}.
   
\subsection{Testing Trend}
The definition of trend in a time series is somewhat nebulous and is discussed in detail in \cite{romano2024permutation}.  Here, we illustrate a particular classical test
of trend that uses 
the Mann-Kendall trend statistic  given by 
	$$
	U_n = \frac{3}{n^{3/2}}\sum_{i < j} \left( I\{X_j > X_i\} - I\{X_j < X_i\} \right) 
	$$
 For  i.i.d. data, $U_n \overset{d} \to N(0, \, 1)$ as $n \to \infty$, and
one can construct exact and asymptotically valid level $\alpha$ tests. 
However, Type 1 error may not be controlled   for stationary sequences
that exhibit no trend.  Once again, studentization can be used to construct asymptotically valid tests, while remaining exact if the underlying process is i.i.d.

\section{Prediction and Conformal Inference}

Permutation tests may, of course,  be used to test a null hypothesis that implies exchangeability of the observations.  Such tests may then be used to construct prediction intervals (or  prediction sets) for future observations. Foundational work for conformal prediction can be found in \cite{vovk2005algorithmic}. See, e.g., \cite{lei2018distribution}, \cite{lei2021conformal}, \cite{barber2021limits}, \cite{barber2021predictive}, and \cite{barber2023conformal} for further developments and applications. \cite{cattaneo2021prediction}, \cite{chernozhukov2021distributional}, \cite{chernozhukov2021exact} give applications of conformal inference to synthetic control (see also \cite{abadie2010synthetic} and \cite{lei2024inference} for alternative approaches closely related to randomization inference). We refer the reader to \cite{angelopoulos2023conformal} for a very nice review of the state of the art of conformal inference.

\subsection{Prediction of an Exchangeable Sequence}

In the simple case, assume $(X_1 , \ldots , X_n , X_{n+1})$ is exchangeable, where the $X_i$ take values in a space ${\cal X}$. It is desired to predict $X_{n+1}$ on the basis
of having observed $X_1 , \ldots , X_n$.
A $100(1- \alpha ) \%$ prediction set $\hat S_n  = \hat S_n ( 1- \alpha, X_1 , \ldots , X_n )$ is a random subset of ${\cal X}$ that satisfies
\begin{equation}\label{equation:predict}
P \{ X_{n+1} \in \hat S_n \} \ge 1- \alpha 
\end{equation}
under all data generating mechanisms $P$ (in an assumed model).
A very general way to construct $\hat S_n ( 1- \alpha )$ is to first
construct a set $E_{n+1} = E_{n+1} (X_1 , \ldots, X_{n+1} )$ in ${\cal X}^{n+1}$ that satisfies
\begin{equation}\label{equation:predict2}
P \{ ( X_1 , \ldots , X_{n+1} ) \in E_{n+1} \} \ge 1- \alpha~~.
\end{equation}
Then, one may construct $\hat S_n$ from $E_{n+1}$ by taking
$$\hat S_n = \{ x : (X_1 , \ldots , X_n , x ) \in E_{n+1} \}~~.$$
Since  $$\{ X_{n+1} \in \hat S_n  \} ~~~{\rm iff}~~~\{ (X_1 , \ldots , X_{n+1} )  \in E_{n+1} \}~,$$
the coverage claim follows.

How can this prescription be applied by using a permutation test?
In words, one may construct $E_{n+1}$ as the acceptance region of a level $\alpha$ test of the null hypothesis that $(X_1 , \ldots , X_{n+1})$
is exchangeable.  Since the randomization hypothesis holds for this null hypothesis using the group of permutations, many such tests are readily available.  

To be concrete, for real-valued observations,  one might base a permutation test
on the test statistic $T_{n+1} =  X_{n+1} - n^{-1} \sum_{i=1}^n X_i $.
(Or, one may wish to consider $|T_{n+1} |$.)
One may view the test statistic as a measure of conformity of $X_{n+1}$ with the remaining data.
Let $\bar X_{n+1} = (n+1)^{-1} \sum_{i=1}^{n+1} X_i$.
Then,  $$T_{n+1} = (1 + \frac{1}{n}) X_{n+1} - \frac{n+1}{n} \bar X_{n+1}~~,$$ so that rejecting for large $T_{n+1}$ is equivalent to rejecting for large $X_{n+1}$ (since $\bar X_{n+1}$ is permutation invariant). 
The acceptance region of (the possibly conservative) nonrandomized permutation test  takes the form (compare with (\ref{equation:accept}))
\begin{equation}\label{equation:accept again}
\{ X_{n+1} : X_{n+1} \le  T^{(k)} (X_1 , \ldots , X_{n+1} ) \}~,
\end{equation}
where $T^{(k)} ( X_1 , \ldots , X_{n+1})$ is the $k$th order statistic among the $n+1$ observations, where $k = \lceil (n+1)(1- \alpha ) \rceil$ and $\lceil c \rceil$ is the smallest integer greater than or equal to $c$. Equivalently, and more in line with the current literature on conformal inference,  the upper bound may be represented as $X_{n, (k)},$ which is the $k$th largest among just $X_1 , \ldots , X_n $. A lower prediction bound can be obtained in the same way by considering $- T_{n+1}$. Alternatively, one can get a two-sided interval based on the test statistic $|X_{n+1} - {\rm med}(X_1 , \ldots , X_n )|$, where ${\rm med}(X_1 , \ldots , X_n )$ is a  median of $X_1 , \ldots , X_n$.  Still another choice replaces the median by the sample mean.

\begin{remark} \rm Note three things. First, the above construction can be viewed as a special case of a two-sample permutation test previously considered based on the two samples $X_1 , \ldots , X_n$ and $X_{n+1}$.
Second, the number of permutations, $(n+1)!$ is reduced  dramatically because one only needs to consider the ${n+1} \choose n$ combinations of distinct values for $T_{n+1}$.  Finally, in the case that all observations are distinct (with probability one), then the region
(\ref{equation:accept again}) has exact coverage $1- \alpha$ if $(n+1) ( 1- \alpha )$ is an integer.
\end{remark}

\subsection{Conformal Prediction Intervals}
More generally, one may wish to predict $Y_{n+1}$ having observed
$(X_1, Y_1 ) , \ldots , (X_n, Y_n )$ and $X_{n+1}$.  The above generalizes almost verbatim.  The goal is to construct a prediction 
region $$\hat S_n = \hat S_n  \left ( 1- \alpha, (X_1 , Y_1 ) , \ldots , (X_n, Y_n ), X_{n+1} \right )$$ that satisfies
\begin{equation}\label{equation:predict3}
P \{ Y_{n+1} \in \hat S_n \} \ge 1- \alpha ~~.
\end{equation}
As in (\ref{equation:predict2}), first  construct $E_{n+1}$  to satisfy
\begin{equation}\label{equationi:predict4}
P \{ ( (X_1, Y_1) , \ldots , ( X_{n+1} , Y_{n+1}  ) \in E_{n+1} \} \ge 1- \alpha~~.
\end{equation}
Then, one may construct $\hat S_n$ from $E_{n+1}$ by taking
$$\hat S_n = \{ y : (X_1 , Y_1 ) , \ldots , (X_n, Y_n ) , (X_{n+1} , y ) \in E_{n+1} \}~~.$$  Then, (\ref{equation:predict3}) is satisfied as before.

First, to construct $E_{n+1}$, as before, simply let $E_{n+1}$ be the
acceptance region of a level $\alpha$ test of the null  hypothesis
that $( (X_1 ,Y_1 ) , \ldots , (X_{n+1} , Y_{n+1} ))$ is exchangeable.
Such tests may be constructed using permutation tests.
For example, suppose $\hat f_{n+1} ( x)  = \hat f_{n+1} ( (X_1 ,Y_1 ) , \ldots , (X_{n+1} , Y_{n+1} ))$ is a generic predictor for an unobserved $Y$ having observed $X = x$, assumed symmetric in its $n+1$  arguments. (Of course, this includes the case the $f_{n+1}$ is specified or obtained from an auxiliary sample.)  Then, one possible test statistic is the fitted residual
$$T_{n+1} = | Y_{n+1} -  \hat f_{n+1} ( X_{n+1} ) |~~.$$
Now one applies a permutation test by recomputing this test statistic,
or what is called a {\it score function} in the conformal prediction literature, over permutations of the data and then finding the $k$th largest value.    

As in the case of just testing exchangeability of $X_1 , \ldots, X_{n+1}$, there are really only $n+1$ permutations or combinations required.  However, the procedure can be computationally quite intensive.  By the duality with the testing problem previously mentioned, in order to determine whether a particular value of $Y_{n+1}$, say $y$,  belongs in the prediction set, one must apply
a permutation test based on the data
$$(X_1 , Y_1 ) , \ldots , (X_n , Y_n ), (X_{n+1} , y)~~.$$
Note that the predictor may change with $y$, and so we write 
$\hat f_{n+1, y}$ for the predictor function, whose dependence on the fixed value $y$ is made explicit.
Let $k = \lceil (n+1) (1- \alpha ) \rceil$ as before.
Define $$t_i = |Y_i - \hat f_{n+1,y} (X_i ) |~.$$
Let $T_y^{(k)}$ be the $k$th largest among $t_1 , \ldots t_n$.
Then, $y$ is included in the prediction region if $$t_{n+1, y  } = |y - \hat f_{n+1,y} (X_{n+1}) | \le T_y^{(k)}~.$$
We emphasize that the construction results in a valid $1- \alpha$ prediction region regardless of whether or not the fitted
$\hat f_{n+1}$ is a reasonable predictor function.  
The key point is that, if  $(X_1, Y_1) , \ldots , (X_{n+1}, Y_{n+1} )$ are assumed exchangeable, then so are    the $n+1 $ random  values  $|Y_i - \hat f_{n+1} ( (X_1, Y_1 ) , \ldots , (X_{n+1} , Y_{n+1} )) |~.$

In order to determine the prediction region, the above computation must be carried out for every $y$. If the range of $y$ values is continuous, then this may require discretization of $y$ values as well. One simple modification is to use a holdout or auxiliary sample to construct the predictor $\hat f$, in which case the computational burden is dramatically reduced.  The usual approach then requires splitting the data, so that $\hat f$ is constructed on one part of the data and the prediction region on the remaining part. See \cite{ritzwoller2023reproducible} for a method that ensures that the residual randomness induced by data splitting is small.

Conformal inference has been generalized from permutation inversion to other types of randomization inference; see \cite{andrews2022transfer}.

\section{Approximate Randomization Tests}

Our analysis so far has considered settings in which the randomization hypothesis holds exactly (often for only a subset of the null hypothesis of interest). We then studied the large-sample behavior of certain tests when the randomization hypothesis does not hold.  In this section, we consider instead settings in which the randomization hypothesis holds only approximately, for large samples, in a sense to be made precise below.  As we will see, such a theory can be fruitfully applied for inference in a variety of complicated settings.  Below, to conserve space, we focus our discussion on regression with clustered data when, in particular, there are a ``small'' number of clusters.  Here, ``small'' means that there are a fixed number of clusters that do not diverge with the sample size. For example, we may measure outcomes for individuals residing in a small number of villages or attending a small number of schools.

This setting is sufficiently general to accommodate a variety of applications -- time series regression, differences-in-differences with a ``small'' number of clusters -- but the underlying idea is considerably more general and can be applied to many other contexts.  See, in particular, \cite{canay2018approximate} for an elegant application to regression discontinuity design.\footnote{\cite{cattaneo2015randomization} also develops a method based on randomization for inference on a regression discontinuity.}

\subsection{Approximate Symmetry}

In order to describe the framework, we assume as before that we observed $X^{(n)} \sim P_n \in \mathbf P_n$ taking values in a sample space $\mathcal X_n$.  We are interested in testing $$H_0 : P_n \in \mathbf P_{n,0}~,$$ where $\mathbf P_{n,0} \subseteq \mathbf P_n$.  The following assumption formalizes the sense in which the randomization hypothesis holds approximately.
\begin{assumption} \label{assumption:art}
Let $S_n : \mathcal X_n \rightarrow \mathcal S$ and $\mathbf G$ be a finite group of transformations $g$ of $\mathcal S$ onto itself.  Suppose $S_n = S_n(X^{(n)}) \stackrel{d}{\rightarrow} S$ under $\{P_n \in \mathbf P_{n,0} : n \geq 1\}$, where $g S \stackrel{d}{=} S$ for all $g \in \mathbf G$.
\end{assumption}
The proposed test mirrors the construction described in Section \ref{sec:generalconstruct}; the key difference is that the role of $X$ there is now played by $S_n$.  Concretely, using the notation described in Section \ref{sec:generalconstruct}, the proposed test is given by  
\begin{equation*}
\phi(S_n) =
\begin{cases}
1 & \text{ if } T(S_n) > T^{(k)}(S_n) \\
a(S_n) & \text{ if } T(S_n) = T^{(k)}(S_n) \\
0 & \text{ if } T(S_n) < T^{(k)}(S_n)
\end{cases}~,
\end{equation*}
where $$a(S_n) = \frac{M\alpha - M^+(S_n)}{M^0(S_n)}~.$$
\cite{canay2017randomization} show that the above test is approximately level $\alpha$ in the sense described by the following theorem:
\begin{theorem} \label{thm:art}
Suppose $\{P_n \in \mathbf P_{n,0} : n \geq 1\}$ satisfies Assumption \ref{assumption:art}; $T : \mathcal S \rightarrow \mathbf R$ is continuous; $g : \mathcal S \rightarrow \mathcal S$ is continuous for all $g \in \mathbf G$; and, for any two distinct $g \in \mathbf G$ and $g' \in \mathbf G$, either $T(gs) = T(g's)$ for all $s \in \mathcal S$ or $P\{T(gS) \neq T(g'S)\} = 1$.  Then, $$E_{P_n}[\phi(S_n)] \rightarrow \alpha$$ as $n \rightarrow \infty$ under $\{P_n \in \mathbf P_{n,0} : n \geq 1\}$.
\end{theorem}
In order to help explain the need for the conditions concerning ties, it is useful to provide some intuition for the theorem.  The key insight is that the randomization test $\phi(S_n)$ only depends on $\{T(gS_n) : g \in \mathbf G\}$ through its ordered values.  To prove the theorem, it therefore suffices to show that when $S_n$ converges to $S$ almost surely (by appealing, e.g., to a suitable almost sure representation theorem) that these ordered values are preserved asymptotically as well.  Other ways of ensuring that this property is satisfied are available; see, e.g., \cite{canay2018approximate} for some such conditions.

\subsection{Clustered Data}

A large class of applications of the preceding machinery share the following structure.  Suppose that 
\begin{equation*} \label{eq:nulltheta0}
\mathbf P_{n,0} = \{P_n \in \mathbf P_n : \theta(P_n) = \theta_0\}~,
\end{equation*}
where $\theta(P_n) \in \mathbf R^d$ is some parameter of interest and $\theta_0$ is some pre-specified value.  Suppose further that the  the data $X^{(n)}$ can be grouped into $q$ clusters $X^{(n)}_1, \ldots, X^{(n)}_q$ and, for $1 \leq j \leq q$, there are estimators $\hat \theta_{n,j} = \hat \theta_{n,j}(X^{(n)}_j)$ using only the data in the $j$th cluster such that $$S_n(X^{(n)}) = \sqrt n (\hat \theta_{n,1} - \theta_0, \ldots, \hat \theta_{n,q} - \theta_0) \stackrel{d}{\rightarrow} N(0,\Sigma)$$ under $\{P_n \in \mathbf P_{n,0} : n \geq 1\}$, where $\Sigma = \text{diag}(\Sigma_1, \ldots, \Sigma_q)$.  When this is the case, it is easy to see that Assumption \ref{assumption:art} is satisfied with $\mathbf G = \{\pm 1\}^q$ and $gs = (g_1, \ldots, g_q)(s_1, \ldots, s_q) = (g_1s_1, \ldots, g_q s_q)$.  Furthermore, \cite{canay2017randomization} show that the remaining requirements of Theorem  \ref{thm:art} are satisfied with $$T_n(S_n) = q \bar S_n ' \hat \Sigma_n^{-1} \bar S_n~,$$ where $\bar S_n = \frac{1}{q} \sum_{1 \leq j \leq q} S_{n,j}$ and $\hat \Sigma_n = \frac{1}{q}\sum_{1 \leq j \leq q} \bar S_{n,j} \bar S_{n,j}'$.  When $d = 1$, \cite{canay2017randomization} show that these conditions are also satisfied for 
\begin{equation} \label{eq:tstatart}
T_n(S_n) = \frac{|\bar S_n|}{\sqrt{\frac{1}{q-1} \sum_{1 \leq j \leq q} (S_{n,j} - \bar S_n)^2}}~.
\end{equation}

We now informally discuss several instances in which the structure described above arises naturally in applications; see \cite{canay2017randomization} for a more formal discussion of these examples.  Of course, the most obvious such example is one in which the grouping of the data is given by the nature of the observations, as in the case of data consisting of observations from villages or schools.  Suppose that the parameter of interest is one or more of coefficients of an ordinary least squares estimand, and that this quantity is invariant with respect to the distribution of data across clusters.  In this case, under assumptions that ensure that dependence within a cluster is sufficiently weak to permit application of suitable central limit theorems and laws of large numbers, and dependence across clusters is sufficiently weak to ensure the diagonal structure of $\Sigma$, the preceding structure is satisfied with $\hat \theta_{n,j}$ being the ordinary least squares estimator using the $j$th cluster of data only.  

In the preceding example, a caveat is that the ordinary least squares estimand must be well-defined within each cluster.  This property may not hold when, e.g., there are variables that are constant within a cluster, as is the case when treatment status is assigned at the level of the cluster.  In this case, it may be necessary to pool one or more clusters together to form even larger clusters, and then apply the above approach to this coarsened clustering of the data.  \cite{canay2017randomization} employ such an approach in the context of differences-in-differences designs.  

Finally, this structure arises naturally in time series settings.  In that case, the clusters can be defined to be ``blocks'' of consecutive observations.  The quantity $q$ must therefore be specified by the user and plays the role of a tuning parameter. Suppose that that the parameter of interest is again one or more coefficients of an ordinary least squares estimand, and that this quantity is invariant with respect to the distribution of data across clusters.  This property would certainly hold under a suitable stationarity assumption.  Under assumptions that again ensure that dependence is sufficiently weak across time, so as to permit application of appropriate central limit theorems and laws of large numbers, the preceding structure is satisfied with $\hat \theta_{n,j}$ being the ordinary least squares estimator using the $j$th cluster of data only.  In particular, the diagonal structure of $\Sigma$ is itself ensured when the dependence across time is sufficiently weak. 

\begin{remark}
Confidence regions may, of course, be constructed using test inversion, as usual.  \cite{cai2023implementation} show that when $d = 1$ and the test statistic in \eqref{eq:tstatart} is employed, the resulting confidence regions are in fact convex, i.e., they are intervals.
\end{remark}

\begin{remark}
In the special case where $d = 1$, the idea of grouping the data in this way and constructing estimators satisfying (\ref{equation:sign1}) has been previously proposed by \cite{ibragimov2010t}.  They propose to test the null hypothesis considered in this subsection by rejecting when $T_n(S_n)$ defined in \eqref{eq:tstatart} exceeds the $1 - \alpha/2$ quantile of the $t$-distribution with $q-1$ degrees of freedom.  By exploiting a remarkable result by \cite{bakirov2006student}, they show that this test has limiting rejection probability no greater than the nominal level under the null hypothesis whenever $\alpha \leq .083$ and $q \geq 2$ or when $\alpha \leq .10$ and $2 \leq q \leq 14$.  Such results, however, are only available for $d = 1$.  Furthermore, the proposed test can be quite conservative under the null hypothesis in the sense that its limiting rejection probability may be far below the nominal level when the $\Sigma_j$ are not all equal.  \cite{canay2017randomization} show that this feature leads the test to be considerably less powerful than the randomization-based test described above in simulations, and further develop some theoretical optimality properties of the randomization-based test in the limiting normal model.
\end{remark}

\begin{remark}
A popular alternative to inference in settings with clustered data is the cluster wild bootstrap proposed by \cite{cameron2008bootstrap}. This approach is especially popular in settings with a ``small'' number of clusters, but to the best of our knowledge the only analysis of its behavior in such settings is contained in  \cite{canay2021wild}, who provide conditions under which the limiting rejection probability under the null hypothesis is no greater than the nominal level in an asymptotic framework like the one described above in which the number of clusters remained fixed with the sample size.  These results are derived by linking its behavior in a fashion similar to the tests described in this section to the behavior of a suitable randomization test based off of the group of sign changes.  Their analysis shows, in particular, that rather restrictive homogeneity conditions on the distribution of covariates across clusters appear to be required for the validity of the cluster wild bootstrap in such settings.
\end{remark}

\newpage
\bibliography{wolf_new}

\end{document}